\documentclass[12pt,a4paper,oneside,openright]{report}
\pdfoutput=1
\usepackage[left=3.5cm,right=2.5cm,top=2.5cm,bottom=3cm]{geometry}

\usepackage[utf8]{inputenc}

\usepackage[T1]{fontenc}

\usepackage[magyar,british]{babel}

\usepackage{comment}

\usepackage{amsmath}

\usepackage{amsthm}
\usepackage{amsfonts}
\usepackage{amssymb}
\usepackage{mathrsfs}

\usepackage{textcomp}
\usepackage{gensymb}

\usepackage[ruled]{algorithm} 
\usepackage{algpseudocode,algorithmicx}
\newcounter{myalgcounter}
\usepackage{tasks}
\usepackage{enumitem}

\usepackage{graphicx}
\usepackage{rotating}
\usepackage{tikz}
\usetikzlibrary{arrows,decorations.markings,patterns,calc, positioning, backgrounds,decorations.pathreplacing}
\definecolor{myred}{RGB}{228,26,28}
\definecolor{mygreen}{RGB}{77,175,74}
\definecolor{myblue}{RGB}{55,126,184}
\definecolor{mypurple}{RGB}{152,78,163}
\definecolor{myorange}{RGB}{255,127,0}

\usepackage{standalone}

\usepackage{nameref}
\usepackage{varioref}

\usepackage[utf8]{inputenc}

\usepackage{tocloft}

\usepackage{caption}
\usepackage{subcaption}

\usepackage{emptypage}

\usepackage[autostyle]{csquotes}

\usepackage{fancyhdr}

\usepackage{tocbibind}

\usepackage{cite}

\usepackage[colorlinks=true,linkcolor=blue,citecolor=blue,filecolor=blue,urlcolor=blue]{hyperref} 
\usepackage{cleveref}
\usepackage{bookmark}

\usepackage[retainorgcmds]{IEEEtrantools}

\newtheorem{theorem}{Theorem}[chapter]
\numberwithin{theorem}{section}

\newtheorem{lemma}[theorem]{Lemma}

\newtheorem{proposition}[theorem]{Proposition}

\theoremstyle{definition} 
\newtheorem*{remark}{Remark}

\theoremstyle{definition}
\newtheorem*{remarks}{Remarks}

\theoremstyle{definition}
\newtheorem{definition}[theorem]{Definition}

\newlength{\subtheoremlength}
\makeatletter
\newcommand*{\defeq}{\mathrel{\rlap{%
			\raisebox{0.3ex}{$\m@th\cdot$}}%
		\raisebox{-0.3ex}{$\m@th\cdot$}}%
	=}

\newcommand*{\eqdef}{=\mathrel{\rlap{%
			\raisebox{0.3ex}{$\m@th\cdot$}}%
		\raisebox{-0.3ex}{$\m@th\cdot$}}}
\makeatother

\ExplSyntaxOn
\NewDocumentCommand{\dceil}{sO{0}m}
{
	\IfBooleanTF{#1}
	{
		\spiros_ceilfloor_ext:nnnn { #2 } { #3 } { \lceil } { \rceil }
	}
	{ 
		\spiros_ceilfloor:nnnn { #2 } { #3 } { \lceil } { \rceil }
	}
}
\NewDocumentCommand{\dfloor}{sO{0}m}
{
	\IfBooleanTF{#1}
	{
		\spiros_ceilfloor_ext:nnnn { #2 } { #3 } { \lfloor } { \rfloor }
	}
	{
		\spiros_ceilfloor:nnnn { #2 } { #3 } { \lfloor } { \rfloor }
	}
}
\cs_new_protected:Npn \spiros_ceilfloor_ext:nnnn #1 #2 #3 #4
{
	\left#3
	\mkern-\muskip_eval:n { 4.5mu + #1mu }
	\left#3
	#2
	\right#4
	\mkern-\muskip_eval:n { 4.5mu + #1mu }
	\right#4
}
\cs_new_protected:Npn \spiros_ceilfloor:nnnn #1 #2 #3 #4
{
	\mathopen{\str_if_eq:nnF { #1 } { 0 } { #1 }#3}
	\mkern\spiros_kern:n { #1 }
	\mathopen{\str_if_eq:nnF { #1 } { 0 } { #1 }#3}
	#2
	\mathclose{\str_if_eq:nnF { #1 } { 0 } { #1 }#4}
	\mkern\spiros_kern:n { #1 }
	\mathclose{\str_if_eq:nnF { #1 } { 0 } { #1 }#4}
}
\cs_new:Npn \spiros_kern:n #1
{
	\str_case:nnF { #1 }
	{
		{ \big }{ -4.5mu }
		{ \Big }{ -5.5mu }
		{ \bigg }{ -6.5mu }
		{ \Bigg }{ -7.5mu }
	}
	{ -4.5mu }
}
\ExplSyntaxOff

\newcommand{\I}{\mathcal{I}}

\newcommand{\R}{\mathbb{R}}
\newcommand{\F}{\mathcal{F}}

\newcommand{\Sch}{\mathscr{S}}

\newcommand{\NPm}{\mathbf{NP}}

\newcommand{\NPh}{$\NPm$-hard}

\newcommand{\Rad}{\mathcal{R}}
\newcommand{\norm}[2][]{\lVert{#2}\rVert_{#1}}

\DeclareMathOperator*{\argmin}{arg\,min}
\DeclareMathOperator{\Span}{Span}
\DeclareMathOperator{\sinc}{sinc}
\DeclareMathOperator\supp{supp}

\DeclareMathOperator{\proj}{proj}

\DeclareMathOperator{\BilInterp}{BilInterp}
\newcommand{\sqpie}{\big( \sqrt{2\pi}\, \big)}
\DeclareMathOperator{\mean}{\mathbb{E}}
\DeclareMathOperator*{\meanlim}{\mathbb{E}}
\DeclareMathOperator{\prox}{prox}

\newcommand{\prob}{\mathbb{P}}

\newcommand{\wei}{\mathcal{W}}
\newcommand{\loss}{\mathcal{L}}
\newcommand{\Loss}{\mathscr{L}}

\newcommand{\logloss}{\mathcal{L}^{\log}}

\input{settings/titlepage_structure.sty}

\graphicspath{{images/}{tikz/}}

\usepackage[textsize=small,textwidth=3cm,disable]{todonotes}
\reversemarginpar
\newcommand{\TODOm}[1]{\todo[color=green!40!white]{TODO: #1}}

\begin{document}

\title{CNN-based regularisation for CT image reconstructions}

\author{Attila Juhos}

\university{Budapest University of Technology and Economics}
\faculty{Faculty of Electrical Engineering and Informatics}
\department{Department of Measurement and Information Systems}
\submitdate{Budapest, May 2021}
\logo{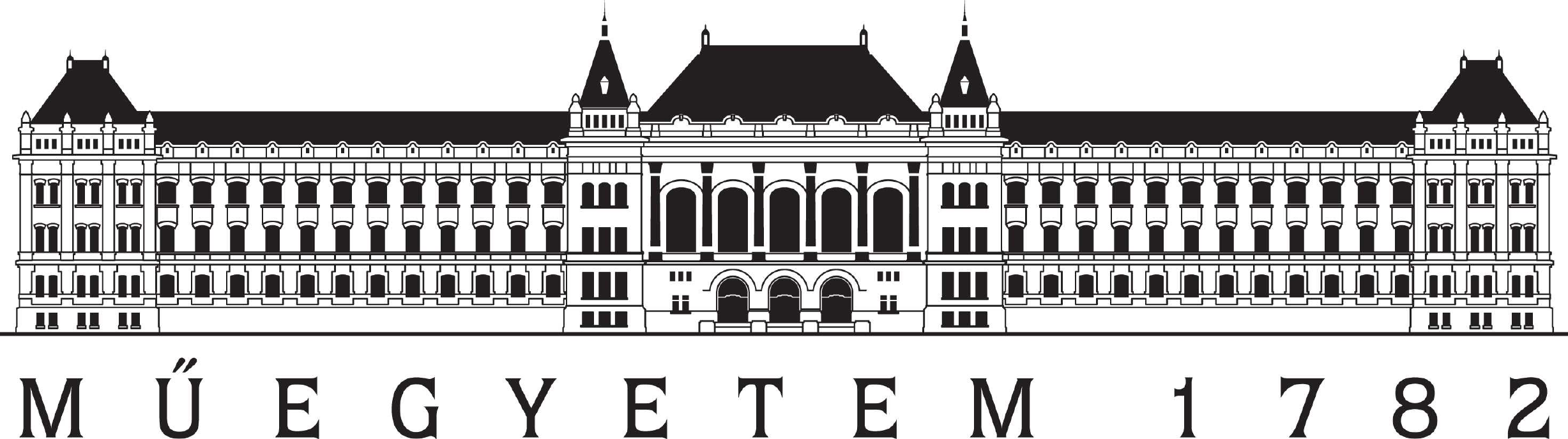}
\supervisor{D\'{a}niel Hadh\'{a}zi (research fellow at BME VIK MIT)}
\supervisorinstitute{BME VIK MIT, Intelligent Systems Research Group}

\maketitle
\preface
\phantomsection\addcontentsline{toc}{chapter}{Abstract}
\begin{abstract}

Computed tomography (CT) involves image reconstruction modalities mostly applied in medical fields. A particular family is represented by X-ray tomographic infrastructures that rely on the acquisition of rays passing through an examined object along with measuring the line integrals of linear attenuation coefficients along such rays. Physical measurements are post-processed by mathematical reconstruction algorithms that may offer weaker or top-notch consistency guarantees on the computed volumetric coefficient field. Superior results are provided on the account of an abundance of low-noise measurements being supplied. Nonetheless, such a scanning process would expose the examined body to an undesirably large-intensity and long-lasting ionising radiation, imposing severe health risks. One main objective of the ongoing research is the reduction of the number of projections while keeping the quality performance stable. Due to the under-sampling, the noise occurring inherently as a consequence of photon-electron interactions is now supplemented by reconstruction artifacts. Nevertheless, we conjecture that the noise distribution applied on the linear attenuation coefficient space is purely dependent on the geometric properties of the scanning. Recently, deep learning methods, especially fully convolutional networks have been extensively investigated and proven to be efficient in filtering such deviations. In this report algorithms are presented that take as input a slice of a low-quality reconstruction of the volume in question and aim to map it to the reconstruction that is considered ideal, the ground truth. Above that, the first system comprises two additional elements: firstly, it ensures the consistency with the measured sinogram, secondly it adheres to constraints proposed in classical compressive sampling theory. The second one, inspired by classical ways of solving the inverse problem of reconstruction, takes an iterative approach to regularize the hypothesis in the direction of the correct result.



\end{abstract}

\body 
\chapter{Introduction} \label{chap:intro}

Computerized Tomography revolves around the well-known integral transformation, the Radon-transformation that is going to be presented in the upcoming sections. What an X-Ray CT scanner is capable of is measuring some values, some instances of the Radon-transform of the scanned object's X-Ray attenuation coefficient space. Naturally, an inverse problem arises. There exist theorems stating how the Radon-transformation could be inverted. For instance, an even theoretically well-established algorithm is the Filtered Back Projection, which is used even nowadays by most commercial X-ray scanning facilities. However, the theory behind always assumes the possibility of recording and processing continuous functions on continuous domains, which is not possible, of course. Besides that, these proofs of correctness count on the capability of providing a sufficient number of projections that guarantee the fulfilment of Shannon's sampling theorem. This condition is unsatisfiable, too. There are other algorithms, the family of algebraic reconstruction techniques (ART, see \cite{Nat01, KS01, AK84}) that naturally implement the discretisation of the Radon-transformation and attempt, in an iterative manner, to reconstruct the volume. Nonetheless, there are proofs that these techniques are not capable of accurately and fully reconstruct the attenuation coefficient space. As a matter of fact, the Radon-transformation itself is a linear operator $\Rad$ mapping the input volume (or slice) $f$ to its sinogram $g$. These iterative approaches usually converge to $\Rad^+g + \proj_{\ker\Rad} f_0$, where $f_0$ denotes the initial guess, $\Rad^+$ is the Moore-Penrose pseudo inverse, and $\proj_{\ker\Rad} f_0$ is the projection of $f_0$ onto the kernel space of $\Rad$. For references, see\cite{AK84, JW03}.

Usual techniques of post-processing noisy reconstructions involve the algorithms suggested by the theory called \emph{compressive sampling} that realised that near-complete reconstruction is possible even with a reduced number of measurements. The techniques build on the fact that the realistic images from the linear space of input images could be embedded in a subspace with far lower dimensionality. This is motivated by the special distribution of input images, for example the well-recognisable lung CT images. \cite{CW08} present an empirical result, according to which they represent a $256 \times 256$ image almost perfectly with the help of only 25 thousand coefficients. Furthermore, linear measurements (like the Radon-transform) are separable with respect to the different pixels of the output image. More precisely, if the Radon-transform $\Rad$ is discretised in a matrix and the $i^{th}$ row is denoted by $\Rad_i$, then one measurement of the slice in question is $\Rad_i f = g_i$. The aforementioned authors also report that the appropriate choice of the measurements and corresponding good representational basis allows us to cut down on a large amount of necessary measurements (from now on called \emph{projections}).

A fully other direction, deterring from the well-established and praised results of CS theory, is consisted by the invocation of the neural paradigm, the use of artificial neural networks, especially fully convolutional networks that contain in their architecture mostly convolutional layers. This approach has become widespread since the appearance of U-Net. In their paper, \cite{RFB15} designed a fully convolutional, encoder-decoder shaped neural network assisted by skip connections. Although they limited their experimentation with the network to biomedical segmentations, various upcoming papers would use more or less the same architecture for other purposes, even outside medical imaging. There are many studies focusing on magnetic resonance imaging (MRI) applications, but there exist much fewer focusing on neural network, especially U-Net aided reconstruction of the linear attenuation coefficient field.

The lines of investigations presented in this work cover two different ideas. The first method, called \emph{measurement-consistent, sparsifying  postprocess-ConvNet,} combines two different approaches proposed recently in the scientific literature. On one side, \cite{HYJ16} argued that artifacts emerging on FBP images using a low number of slices follows a distribution dependent on the geometric arrangement of scanning and, therefore, believe that the reconstruction network should learn the artifacts, i.e.\ it should learn the differences between the ideal and the fed reconstruction. On the other side, based on \cite{HP19} we propose a U-Net-based network, which besides prescribing fidelity to the ideal reconstruction, trains parameters in a way that consistency of the output with the measured projections is provided throughout the process. More precisely, the output is Radon-transformed and loss is generated based on the inconsistencies between the expected sinogram and the predicted one. Additionally, constraints from the field of compressed sensing are employed to regularise learning in a way, that output reconstructions are sparse in the sense of total variation. At the end of the report we prove that this approach successfully optimises the denoising capacity of the network. As far as our literacy extends, this is the first time that a combined, hybrid system has been devised. 

Another line of investigation of this report developed a fully iterative scheme, called \emph{unrolled support-kernel iterative regularisation GD}, which alternates between applying steps of iterative algebraic steps and an iteratively taught neural network. Our idea was based mostly on the work of \cite{GJ18}, where the authors created an algorithm called \emph{projected gradient descent (PGD)}. There, the reconstruction is alternating between algebraic steps and a projection to the manifold of reconstructions. This projection is provided by a convolutional neural network. Our addition and major change compared to this method is that the neural network becomes part of the iterative refinement system and produces iterative kernel space reconstruction steps on the current hypothesis.

The structure of the report is the following. In Chapter~\ref{chap:preliminaries} and Chapter~\ref{chap:cnn} the mathematical aspects of medical computing, CT and neural networks are presented in detail, respectively. The latter also contains an introduction to solving inverse problems via CNNs and previous results. Chapters~\ref{chap:convnet} and \ref{chap:unrolled} present the main novelties of this report, alongside with experimental results. Finally, in Chapter~\ref{chap:conclusion} we conclude the report and state various open questions along with future research directions.

\chapter[Mathematical background of CT]{Mathematical background of Computerised Tomography} \label{chap:preliminaries}

Computerised tomography (CT) aims to present the inner structure of a body via a representative function of the space. A major field of application of the theory is obviously the CT based medical imaging, where it is desirable to somehow record the inner structure of the patient's body in such a quality that the received representation could be used for diagnostics, most often cancerous tumour detection. This underlying representation is provided by an objective function defined on the subset of the 3D-space. The meaning of the function reflects a physical property that was measured by the CT scanning modality. For example, in X-ray based CT the objective function is the linear attenuation coefficient function of the body, which expresses to what extent that point of the body attenuates a passing X-ray. This measure intends to model the physical interactions between X-ray photons and electrons. The modelling of all processes, unfortunately, would require the probabilistic modelling of physical interactions, which are dependent on the energy distribution of the incident ray, the molecules present in the tissue, the spatial and energy distribution of emitted electrons or photons after a colliding interaction. 

The modelling used by us and by most of CT reconstruction algorithms neglects all these factors in favour of receiving a model that results in a linear system description, tractable numerical simulations and algorithms favouring current capacities of computers. The process of attenuating a ray is, hence, described in a very simplified attempt by the Beer-Lambert-law:
\begin{equation*}
I = I_0 e^{-\int_{L} \mu(x)dx},
\label{eq:Beer-Lambert}
\end{equation*} 
where $L$ represents the path the ray traverses and $\mu$ is attenuation function. As we see, after taking logarithms, this simplified description of the attenuation results in a linear system operator and opens up the treatment of computerised tomography to the extensive arsenal of analysis and solution methods of linear inverse problems:
\begin{equation*}
\int_L \mu(x) dx = -\ln\frac{I}{I_0}.
\end{equation*}
Upcoming theorems also yield a direct inversion formula for a complete availability of line integrals. Nevertheless, a practical measurement acquisition setup fixes a set of finite number of line integrals, called \emph{scanning geometry}. Besides the natural discreteness of data, incompleteness may occur in other forms. Firstly, it may be the result of exposing only limited regions of the examined body to the otherwise undesirably ionising radiation. Secondly, implants generally absorb rays and the particular line integral cannot be reliably measured. Furthermore, sometimes the measurements are only available in a restricted angular range rather than a full circle (technically a half circle). In these cases we speak about incomplete data. 

This report will focus on the simplest scanning geometry involving equidistant parallel beams under angles evenly distributed throughout the entire angular range. We refer to this as \emph{parallel scanning geometry}, see Fig.~\ref{fig:sample_geometries:b}. In this acquisition infrastructure the body lies between a sequence of radiation sources and a detector panel. The sources and the detectors rotate simultaneously around the domain of interest. We note that usual medical CT infrastructures do have a full angular range, however almost always a \emph{fan-beam} scanning geometry (Fig.~\ref{fig:sample_geometries:c}) is applied, where only one radiation source is used. For a sample reconstruction see Fig.~\ref{fig:sample_geometries:a}.

\begin{figure}[ht]
\centering
\begin{minipage}{.3\linewidth}	
	\centering
	\subfloat[]{
		\includegraphics[keepaspectratio,width=\linewidth]{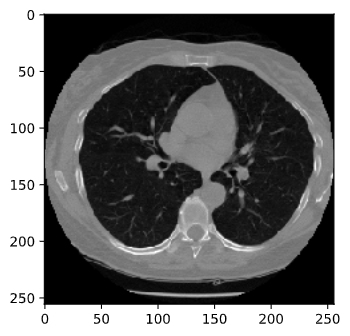}
		\label{fig:sample_geometries:a}
	}
\end{minipage}
\begin{minipage}{.3\linewidth}
	\centering
	\subfloat[]{
		\label{fig:sample_geometries:b}
		\includegraphics[keepaspectratio,width=\linewidth]{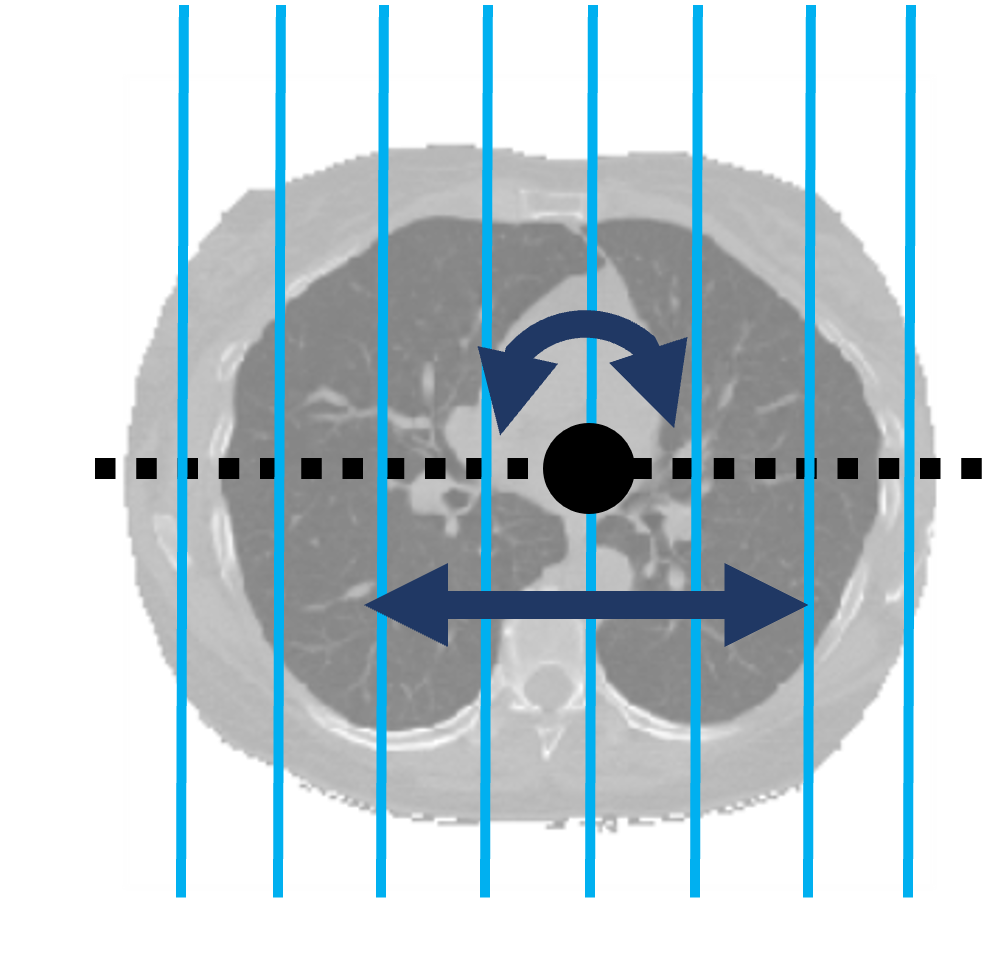}			
	}
\end{minipage}
\begin{minipage}{.3\linewidth}
	\centering
	\subfloat[]{
		\label{fig:sample_geometries:c}
		\includegraphics[keepaspectratio,width=\linewidth]{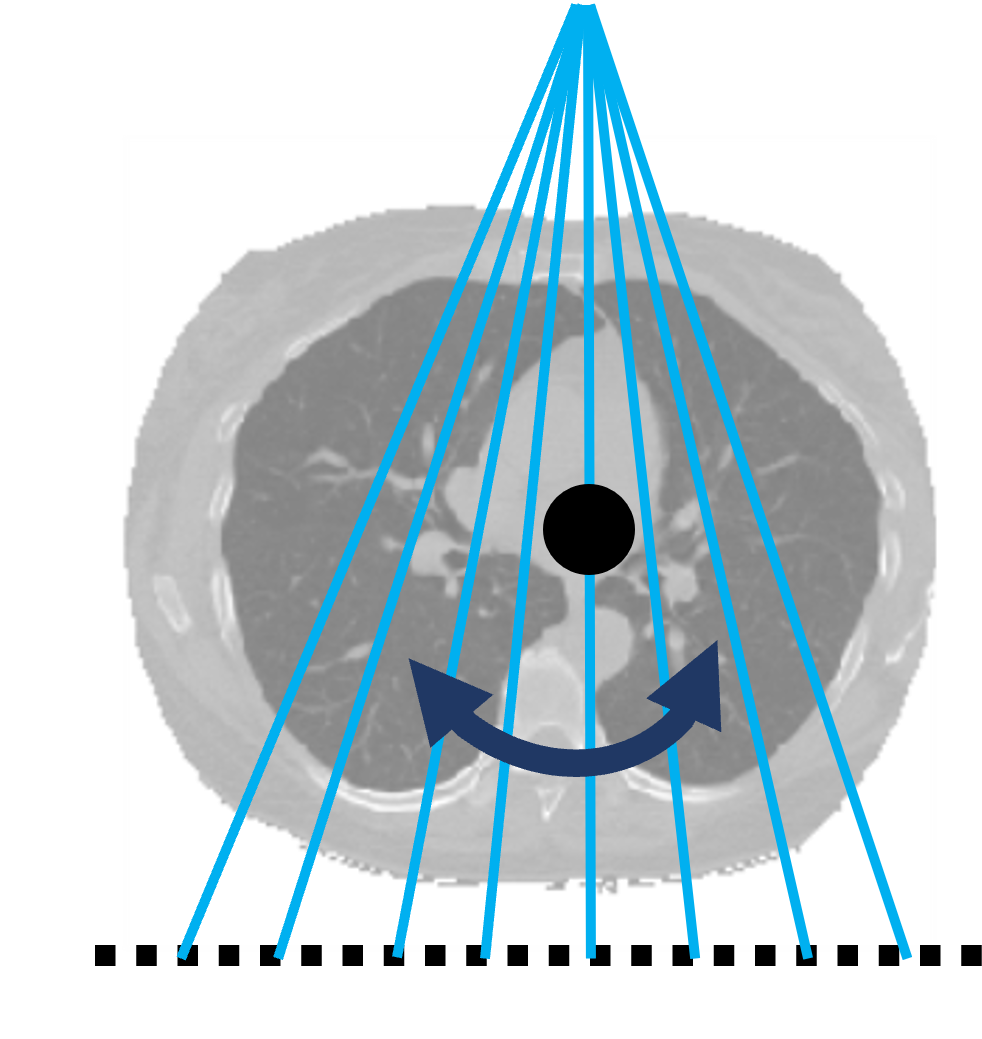}			
	}
\end{minipage}

\caption{\subref{fig:sample_geometries:a} A sample reconstruction. \subref{fig:sample_geometries:b} Parallel scanning geometry. \subref{fig:sample_geometries:c} Fan-beam scanning geometry.}
\label{fig:sample_geometries}
	
\end{figure}



In this chapter we present the mathematical approach of CT alongside the most important reconstruction algorithms. As a source material for this chapter, we used and highly recommend the textbooks \cite{Nat01}, \cite{NW01}, \cite{Hel80}, \cite{KS01} and the lecture notes \cite{BGJ20}. Section~\ref{sec:Radon} introduces the notion of Radon-transformation, immediately followed by several important properties in section~\ref{sec:properties_Radon}. Amongst others we provide and prove the inversion formula. Afterwards, discretisation of data and the Radon-transform are outlined in section~\ref{sec:discretisation}. Section~\ref{sec:discretised_FBP} derives the most widespread reconstruction algorithm based on the continuous inversion of the Radon-transform, alongside some theoretical guarantees of exactness. Last, but not least, section~\ref{sec:art} displays a handful of algebraic methods for solving linear problems.

The statement and proof of shift-invariance in the discretised case, discussed in subsections~\ref{subsec:disc_Radon} and \ref{subsec:discrete_LI}, was derived by us in an attempt to fill the gap between the corresponding continuous statements and discrete reconstruction algorithms.

\section{Radon-transformation} \label{sec:Radon}
Starting from the Beer-Lambert-law:
\[I = I_0 \cdot e^{-\int_L \mu(x) dx},\]
we have already derived (by taking logarithms, step often called \emph{linearisation}) the linearised form:
\begin{equation*}
\int_L \mu(x) dx = -\ln\frac{I}{I_0},
\end{equation*}
which means that after the measurements, by division with the initial intensity we obtain the line integral of the attenuation coefficient. This gives rise to the definition of the Radon-transform. If $f$ is a function defined on $\R^n$, then for any $\Theta \in S^{n-1}$ unit vector (where $S^{n-1}$ denotes the unit ball of $\R^n$) and any $s \in \R$ we define:
\begin{equation}
\Rad f(\Theta, s) = \Rad_\Theta f(s)
\defeq \int_{\langle\Theta, x\rangle=s} f(x) dx 
= \int_{\text{Span}\{\Theta\}^{\perp}} f(s\Theta + x) dx,
\label{eq:Radon-def}
\end{equation}
where $\langle\cdot,\cdot\rangle$ is the standard Euclidean inner product defined on $\R^n$, $\text{Span}\{\Theta\}^{\perp}$ is the perpendicular subspace to the spanned subspace of $\Theta$. Hence, the Radon-operator, as an integral transformation operator, maps functions defined on $\R^n$ to functions defined on $\mathcal{Z} = S^{n-1}\times\R \subseteq \R^{n+1}$. The resulting function is called the \emph{sinogram} or the \emph{Radon-transform} of $f$. The partial function $\Rad_\Theta f$ is called a \emph{projection} belonging to direction $\Theta$. 


It should be noticed that this integral transformation takes integrals on hyperplanes of $\R^n$ rather than lines (i.e.\ one-dimensional subspaces). This is motivated by the fact that X-ray CT modalities are usually implemented in such a way that the measurements are taken slice by slice. More precisely, in \emph{computed axial tomography} the detector and the X-ray source rotates around an axial slice of the examined volume. In case of helical CT, we interpolate the received data to the same representation. Hence, we only get information from one plane, where hyperplanes coincide with straight lines and, thus, the theory on hyperplane integrals is applicable.

As far as its existence is concerned, the integral in~\eqref{eq:Radon-def} is well defined when $f$ is an element of the Schwartz-space $\Sch(\R^n)$. Schwartz-spaces are special linear function spaces that contain smooth, i.e.\ infinite times differentiable functions that converge to $0$ by distancing from the origin ``sufficiently fast''. More precisely, for any $\Omega \subset \R^{n}$ open set the Schwartz-space on $\Omega$ is defined as
\[ \Sch(\Omega) = \big\{ f:\Omega \rightarrow \R \text{ smooth}\,\big|\, \text{for any } \alpha, \beta \text{ multi-indices, } \sup_{x \in \Omega} |x^\beta \partial^\alpha f(x)| < \infty \big\}. \]

On the other hand, if $f \in \Sch(\R^n)$, then it is provable that $\Rad_\Theta f \in \Sch(\R)$ and $\Rad f \in \Sch(\mathcal{Z})$, where $\Sch(\mathcal{Z})$ is simply the restriction of $\Sch(\R^{n+1})$ to $\mathcal{Z}$. As for a trivial property, the linearity of integrals in terms of its integrands is inherited by both $\Rad_\Theta$ and $\Rad$, therefore we conclude, that $\Rad_\Theta \in Lin(\Sch(\R^{n}), \Sch(\R))$ and $\Rad \in Lin(\Sch(\R^{n}), \Sch(\mathcal{Z}))$.

It should be stated that Schwartz-spaces contain any smooth function that vanish outside a bounded subset of the domain. This is convenient in the medical X-ray tomography, since the attenuation coefficient of air outside the examined body could be considered zero and, therefore, the only assumption in this regard is having a smooth attenuation within the body.

Finally, it should be clear that from a computer scientist's perspective our main goal is to somehow invert the Radon-transform, i.e.\ given $\Rad f$, find $f$. First, a few essential properties of the Radon-transform are outlined.

\section{Properties of the Radon-transformation} \label{sec:properties_Radon}
In this section we enlist a key selection of properties of the Radon-transformation that will be important in the understanding of reconstruction algorithms. Here we repeat the fact that the Radon-transformation is well-defined on Schwartz-functions and is linear both angle-wise and in general: $\Rad_\Theta \in Lin(\Sch(\R^{n}), \Sch(\R))$ and $\Rad \in Lin(\Sch(\R^{n}), \Sch(\mathcal{Z}))$.

\subsection{Fourier slice theorem}

\begin{theorem}[Fourier slice theorem]
\label{theo:slice}
For any $f \in \Sch(\R^n)$ the Fourier-transform of $\Rad_\Theta$ exists and for any $\sigma \in \R$:
\[ \F \{\Rad_\Theta f\}(\sigma) = \big(\sqrt{2\pi}\,\big)\!^{n-1} \F f(\sigma\Theta). \]
\end{theorem}

\begin{proof}
Since $\Rad_\Theta f \in \Sch(\R)$, its Fourier-transform exists. Then
\[ \F\{\Rad_\Theta f\}(\sigma) = \frac{1}{\sqrt{2\pi}} \int_\R e^{-i\sigma s} \Rad_\Theta f(s) ds 
= \frac{1}{\sqrt{2\pi}} \int_\R e^{-i\sigma s} \int_{\Span\{\Theta\}^\perp} f(s\Theta + y) dy \, ds. \]
Bringing the exponential term inside the inner integral, we see that while $s$ traverses $\R$ and $y$ traverses $\Span\{\Theta\}^\perp$, the expression $s\Theta + y$ fills $\R^n$. Hence, by substituting $x = s\Theta + y$ and noticing that $\langle x, \Theta \rangle = s$, the following is derived:
\begin{align*} 
\F\{\Rad_\Theta f\}(\sigma) & = \frac{1}{\sqrt{2\pi}} \int_\R \int_{\Span\{\Theta\}^\perp} e^{-i\sigma s} f(s\Theta + y) dy \, ds\\
& = \big(\sqrt{2\pi}\,\big)\!^{n-1} \cdot \frac{1}{\big(\sqrt{2\pi}\,\big)\!^{n}} \int_{\R^n} e^{-i \langle x, \sigma \Theta \rangle} f(x) dx = \big(\sqrt{2\pi}\,\big)\!^{n-1} \F f(\sigma\Theta).\quad \qedhere
\end{align*}
\end{proof}

The Fourier slice theorem's meaning is that the spectrum of a single projection is proportional exactly to the $\Theta$-direction slice of the multidimensional spectrum of $f$.
Some textbooks even define the Radon-transform on direction $\Theta$ by taking the Fourier-inverse of this slice of $\F f$, arriving to an equivalent definition.

\subsection{Adjoint operators}

We are about to define the adjoint operators of both $\Rad_\Theta$ and $\Rad$ operators. For this it is noticed that any Schwartz-space $\Sch(\Omega)$ is a subset of $L^2(\Omega)$ and, hence, they are equipped with the usual inner product defined as
\begin{equation*}
\langle f, g \rangle_{\Sch(\Omega)} = \int_\Omega fg.
\end{equation*}
Thus, the operators $\Rad_\Theta \in Lin(\Sch(\R^{n}), \Sch(\R))$ and $\Rad \in Lin(\Sch(\R^{n}), \Sch(\mathcal{Z}))$ are defined between inner product spaces.

\begin{theorem}
\label{theo:adjoints}
The operators $\Rad_\Theta$ and $\Rad$ admit adjoint operators which are defined as:
\begin{enumerate}
\item $\Rad_\Theta^* \in Lin(\Sch(\R), \Sch(\R^n)), \Rad_\Theta^* g(x) = g(\langle x, \Theta \rangle)$ and
\item $\Rad^* \in Lin(\Sch(\mathcal{Z}), \Sch(\R^n)), \Rad^*g(x) = \int_{S^{n-1}}g(\Theta, \langle x, \Theta \rangle)d\Theta$, where $S^{n-1}$ denotes, again, the unit ball of $\R^n$.
\end{enumerate}
\end{theorem}

\begin{remark}
There is a fairly intuitive interpretation of adjoint $\Rad^*$. The direct measurement operator $\Rad$ takes projections orthogonally to $\Theta$ at various distances $s$ from the origin. This is the reason why $\Rad$ is called \emph{forward projection (operator)} in a CT context. Let us examine the contribution of a point $x \in \R^n$ to the different projections in the Radon-transform. For every single direction $\Theta$, $x$ is found on the hyperplane orthogonal to $\Theta$ at distance $\langle x, \Theta \rangle$. Consequently, what $\Rad^* (\Rad f)$ does is that it ``backprojects'' values from $\Rad f$ from the locations that $x$ has contributed to in $\Rad f$. Hence, in the CT terminology $\Rad^*$ is always referred to as \emph{backprojection (operator)}. \TODOm{an illustration about backprojection could be nice}
\end{remark}

\begin{proof}[Proof of Theorem~\ref{theo:adjoints}]
Let $f \in \Sch(\R^n)$ and $g \in \Sch(\R)$. We use a substitution similar to the one applied in the proof of Theorem~\ref{theo:slice}:
\begin{align*}
\langle g, \Rad_\Theta f \rangle_{\Sch(\R)} &= \int_\R g(s)\Rad_\Theta f(s)ds = \int_\R \int_{\Span\{\Theta\}^\perp} g(s)f(s\Theta + y)dy\,ds\\
&= \int_{\R^n} g(\langle x, \Theta \rangle) f(x) dx = \langle g\left(\langle \cdot, \Theta \rangle\right), f \rangle_{\Sch(\R^n)},
\end{align*}
hence the first statement. For the second one let $h \in \Sch(\mathcal{Z})$.
\begin{align*}
\langle h, \Rad f \rangle_{\Sch(\mathcal{Z})} &= \int_{S^{n-1}} \int_\R g(\Theta, s) \Rad f(\Theta, s) ds\,d\Theta \\
&= \int_{S^{n-1}} \int_\R \int_{\Span\{\Theta\}^\perp} g(\Theta, s) f(s\Theta + y)dy\,ds\,d\Theta \\
&= \int_{S^{n-1}} \int_{\R^n} g(\Theta, \langle x, \Theta \rangle ) f(x)dx\,d\Theta \\
&= \int_{\R^n} \left( \int_{S^{n-1}}  g(\Theta, \langle x, \Theta \rangle )d\Theta \right) f(x)dx 
= \left\langle \int_{S^{n-1}}  g(\Theta, \langle \cdot, \Theta \rangle )d\Theta, f \right\rangle_{\Sch(\R^n)}.
\end{align*}
\end{proof}

Here we restate without proof the analogue of the slice theorem for the adjoint operator. It again clarifies the connection between the spectrum of the backprojected function and the original spectrum:
\begin{theorem}
\label{theo:slice_adjoint}
For $g \in \Sch(\mathcal{Z})$ even, i.e.\ $g(-\Theta, -s) = g(\Theta, s)$, and any $\xi \in \R^n$, we have:
\[ \F\Rad^* g(\xi) = 2 \cdot \big(\sqrt{2\pi}\, \big)^{n-1} \norm[2]{\xi}^{-n+1} \F g\left(\frac{\xi}{\norm[2]{\xi}}, \norm[2]{\xi}\right).  \]
\end{theorem}

\remark{Keep in mind, that any Radon-transform $\Rad f$ is an even function of $\Sch(\mathcal{Z})$.}

\subsection{Convolutional properties}

For the next property, the convolution of two Radon-transforms has to be defined. From now on the convolution of two functions on $\mathcal{Z} = S^{n-1}  \times \R$ is defined as the convolution in their second variable. More precisely, if $g, h \in \Sch(\mathcal{Z})$, then the convolution function defined below is also in $\Sch(\mathcal{Z})$:
\[ (g \ast h)(\Theta, s) = \int_\R g(\Theta, s-t)h(\Theta, t)dt. \]

\begin{theorem}
\label{theo:convolutional}
For any $f \in \Sch(\R^n), g \in \Sch(\mathcal{Z})$, the following holds:
\begin{equation*}
\Rad^*g \ast f = \Rad^*\{g \ast \Rad f\}.
\label{eq:convolutional}
\end{equation*}
\end{theorem}

\remark{Theorem~\ref{theo:convolutional} will be our starting point to develop a reconstruction algorithm. We will be looking for filters $g$, for which $\Rad^* g$ is almost a Dirac-function.}

\begin{proof}[Proof of Theorem~\ref{theo:convolutional}]
We have:
\begin{align*}
(\Rad^*g \ast f)(x) &= \int_{\R^n} \left( \int_{S^{n-1}} g(\Theta, \langle x-y, \Theta \rangle)d\Theta \right) f(y) dy\\
&= \int_{S^{n-1}} \int_{\R^n} g(\Theta, \langle x, \Theta \rangle - \langle y, \Theta \rangle) f(y) dy \, d\Theta.
\end{align*}
After substituting $y=s\Theta + z$, where $s\in\R$ and $z \in \Span\{\Theta\}^\perp$, we get:
\begin{align*}
(\Rad^*g \ast f)(x) &= \int_{S^{n-1}} \int_\R \int_{\Span\{\Theta\}^\perp} g(\Theta, \langle x, \Theta \rangle - s) f(s\Theta +z) dz \, ds \, d\Theta \\
&= \int_{S^{n-1}} \int_\R g(\Theta, \langle x, \Theta \rangle - s) \left( \int_{\Span\{\Theta\}^\perp} f(s\Theta +z) dz \right) \, ds \, d\Theta
\end{align*}
Thus,
\begin{align*}
(\Rad^*g \ast f)(x) &= \int_{S^{n-1}} \left(\int_\R g(\Theta, \langle x, \Theta \rangle - s) \Rad f(\Theta, s) ds \right) d\Theta \\
&= \int_{S^{n-1}} (g \ast \Rad f)(\Theta, \langle x, \Theta \rangle) d\Theta = \Rad^*\{ g\ast\Rad f \}(x).  &\qedhere
\end{align*}
\end{proof}

\begin{theorem}
\label{theo:continuous_LI}
For any $f \in \Sch(\R^n)$ we have:
\begin{equation*}
\Rad^*\Rad f = |S^{n-2}| \cdot \norm[2]{\I_{\R^n}}^{-1} \ast f,
\label{eq:continuous_LI}
\end{equation*}
where $\I_\Omega$ denotes the identity operator of $\Omega$ and $|S^{n-2}|$ is the surface of the unit ball in $\R^{n-1}$. For $n=2$ we have $|S^{0}| = 2$.
\end{theorem}
\TODOm{check whether this exact formulation holds for 2d and 3d. By exact I mean the constant}

\remark{This theorem is the certificate of the fact that $\Rad^*\Rad$, i.e.\ a combined forward- and backward-projection is a \emph{linear, shift-invariant (LI) operator}. This operator will be a key element in some iterative reconstruction schemes.}

\begin{proof}[Proof of Theorem~\ref{theo:continuous_LI}]
For the proof, the following lemma from \cite[VII. 2.]{Nat01} will be handy:
\begin{equation}
\int_{S^{n-1}} \int_{\Span\{\Theta\}^\perp} f = |S^{n-2}| \int_{\R^n} \norm[2]{\I_{\R^n}}^{-1}f.
\label{eq:lemma_for_continuousLI}
\end{equation}
Obviously $\Rad^*\Rad \in Lin(\Sch(\R^n), \Sch(\R^n))$.
\begin{align*}
\Rad^*\Rad f(x) = \int_{S^{n-1}} \int_{\Span\{\Theta\}^\perp} f(\langle x, \Theta \rangle \Theta + y)dy \, d\Theta.
\end{align*}
However, $\langle x, \Theta \rangle \Theta + \Span\{\Theta\}^\perp = x + \Span\{\Theta\}$, therefore:
\begin{align*}
\Rad^*\Rad f(x) = \int_{S^{n-1}} \int_{\Span\{\Theta\}^\perp} f(x + y)dy \, d\Theta.
\end{align*}
After appyling~\eqref{eq:lemma_for_continuousLI} for $y \rightarrow f(x+y)$ and substituting $y = z-x$, it follows:
\begin{equation*}
\Rad^*\Rad f(x) = \int_{\R^n} \norm[2]{x-z}^{-1}f(z)dz = (\norm[2]{\I_{\R^n}}^{-1} \ast f)(x).\qedhere
\end{equation*}
\end{proof}

\subsection{Inverse. Continuous FBP} \label{subsec:continuous_FBP}

For the inverse operator we introduce a new notation called the Riesz-potential, which denotes a special type of filter using a power of $|\I_{\R}|$ in Fourier domain. For $h$ defined on $\R$:
\[ I^\alpha h \defeq \F^{-1}\{ |\I_{\R}|^{-\alpha} \F h \}. \]
From now on, likewise the convolution operator, the Fourier-transformation and Riesz-potential operators act on the second variable of functions defined on $\mathcal{Z}=S^{n-1}\times\R$. More precisely, for $g \in \Sch(\mathcal{Z})$:
\[ \F g (\Theta, \cdot) = \F \{ s \rightarrow g(\Theta, s) \}(\cdot) \text{ and } I^\alpha g = \F^{-1} \{ s \rightarrow |s|^{-\alpha}\F g(\Theta, s) \}. \]

\begin{theorem}
\label{theo:continuous_inverse}
For $f\in\Sch(\R^n)$, the following makes sense and holds:
\begin{equation}
f = \frac{1}{2} \cdot \frac{1}{(2\pi)^{n-1}} \Rad^* I^{-n+1}\Rad f.
\label{eq:continuous_inverse}
\end{equation}
Therefore the operator $\Rad$ is left invertible and $\Rad_l^{-1} = \frac{1}{2} \cdot \frac{1}{(2\pi)^{n-1}} \Rad^* I^{-n+1}$.
For $n=2$ we have (by abuse of notation):
\begin{equation}
f = \frac{1}{2} \cdot \frac{1}{2\pi} \int_{S^1} \F^{-1} \{ |\I_\R| \F \Rad f \}(\Theta, \langle x, \Theta \rangle)d\Theta.
\label{eq:continuous_FBP_2D}
\end{equation}
\end{theorem}

\begin{remark}
The inverse Radon-operator is often referred to as \emph{filtered backprojection (FBP)}, and in our case we derived the continuous form of FBP. In the two-dimensional case, $I^{-1}g = \F^{-1} \{ |\I_\R|\F g \}$ is called ramp-filtering and $|\I_\R|$ is the \emph{ramp-filter}. Note again that $|\I_\R|$ acts on the second variable of $\F \Rad f$, i.e.\ filtering happens \emph{projection-wise}.
\end{remark}

\begin{proof}[Proof of Theorem~\ref{theo:continuous_inverse}]
The Fourier inversion formula for $f$:
\[ f(x) = \frac{1}{\big( \sqrt{2\pi} \,\big)^n } \int_{\R^n} e^{i\langle x, \xi \rangle } \F f(\xi) d\xi. \]

We subsitute polar coordinates. In general, if $\xi = \sigma \Theta$, where $\sigma \in [0, +\infty)$ and $\Theta \in S^{n-1}$, then
\[ \int_{\R^n} h(\xi) d\xi = \int_{S^{n-1}} \int_0^{\infty} \sigma^{n-1} f(\sigma\Theta) d\sigma \, d\Theta. \]
Therefore,
\[ f(x) = \frac{1}{\big( \sqrt{2\pi} \,\big)^n } \int_{S^{n-1}} \int_0^{\infty} e^{i \sigma \langle x, \Theta \rangle} \sigma^{n-1} \F f(\sigma\Theta) d\sigma \, d\Theta.  \]

By substituting $(-\Theta, -\sigma)$ in the place of $(\Theta, \sigma)$, we get
\[ f(x) = \frac{1}{\big( \sqrt{2\pi} \,\big)^n } \int_{S^{n-1}} \int_{-\infty}^{0} e^{i \sigma \langle x, \Theta \rangle} (-\sigma)^{n-1} \F f(\sigma\Theta) d\sigma \, d\Theta.  \]

The terms $\sigma^{n-1} |_{\sigma \geq 0}$ and $(-\sigma)^{n-1} |_{\sigma \leq 0}$ are merged into $|\sigma|^{n-1}$. By summing up the two expression, we obtain:
\begin{equation}
f(x) = \frac{1}{2} \cdot \frac{1}{\big( \sqrt{2\pi} \,\big)^n } \int_{S^{n-1}} \int_{-\infty}^{\infty} e^{i \sigma \langle x, \Theta \rangle} |\sigma|^{n-1} \F f(\sigma\Theta) d\sigma \, d\Theta.
\label{eq:continuous_inverse_parital1}
\end{equation}

By applying the Fourier slice theorem~\ref{theo:slice}, we get:
\begin{align*}
f(x) &= \frac{1}{2} \cdot \frac{1}{\big( \sqrt{2\pi} \,\big)^{2n-1} } \int_{S^{n-1}} \int_{\R^n} e^{i \sigma \langle x, \Theta \rangle} |\sigma|^{n-1} \F \Rad f(\Theta, \sigma) d\sigma \, d\Theta \\
&= \frac{1}{2} \cdot \frac{1}{\big( \sqrt{2\pi} \,\big)^{2n-2} } \int_{S^{n-1}} \F^{-1} \{ |\I_\R|^{n-1}\F \Rad f\}(\Theta, \langle x, \Theta \rangle) d\Theta,
\end{align*}
from which the expected result is concluded.
\end{proof}

\begin{remarks}\leavevmode
\begin{enumerate}
\item From \eqref{eq:continuous_inverse_parital1} we could have derived another formula. Let $H^{n-1}$ be a half sphere, e.g.\ in a two dimensional case the unit vectors corresponding to the $[0, \pi]$ angular range. By decomposing the outer integral into separate integrals on $H^{n-1}$ and $-H^{n-1}$, we observe that after substituting in one of them again $(-\Theta, -\sigma)$, the terms become equal.
Hence, in this case, for the two dimensional case we arrive to:
\begin{equation}
f = \frac{1}{2\pi} \int_{H^1} \F^{-1} \{ |\I_\R| \F \Rad f \}(\Theta, \langle x, \Theta \rangle)d\Theta.
\label{eq:continuous_inverse_half}
\end{equation}
This is important, because most implementations follow this formula. After interpretation, we immediately realise that considering every single projection twice at both $(\Theta, s)$ and $(-\Theta, -s)$ is unnecessary.

\item We only proved the left side invertibility of $\Rad$. In fact, $\Rad : \Sch(\R^n) \rightarrow \Sch(\mathcal{Z})$ is not surjective, since for any $f \in \Sch(\R^n)$ it is trivial that $\Rad f$ is even, $\Rad f(-\Theta, -s) = \Rad f(\Theta, s)$. If the target domain is restricted to even functions of $\Sch(\mathcal{Z})$, then the $\Rad$ operator becomes fully invertible and, obviously, $\Rad^{-1} = \Rad_l^{-1}$.

\end{enumerate}
\end{remarks}

\subsection{Filtering}
In this subsection we discuss the design of filtering functions for the FBP algorithm. Even though the ramp-filter was already introduced in subsection~\ref{subsec:continuous_FBP} alongside a direct inversion formula, the discussion is rather started off from Theorem~\ref{theo:convolutional}: $\Rad^*v \ast f = \Rad^* \{ v \ast \Rad f \}$. Analysing this equation provides the benefit of directly designing other filter functions. The ideal situation would be having $\F v$ as the ramp-filter, which would result in $\Rad^* v$ being the Dirac $\delta$. The drawback of this approach is two-folded: firstly, the inverse Fourier of the ramp-filter is not easily computed, and secondly, it would amplify high-frequency noise components in the projections. We, therefore, aim for faithful reconstruction of $\Omega$-band-limited functions $f$ (i.e~$\F f(\xi)=0$, if $\norm[2]{\xi} > \Omega$). Watch out that due to the slice theorem, both the projections and the Radon-transform of $f$ become, in that case, $\Omega$-band-limited. Hence, instead of having $\F \Rad^*v = \F \delta = 1/\sqpie^{n}$, we choose $\F \Rad^*v = 1/\sqpie^{n}$ on a $\Omega$-radius support (and $0$ otherwise). Or more generally, a \emph{filter factor} $\Phi$ is introduced such that $\F\Phi(\sigma)$ is close to $1$ if $|\sigma| \leq 1$ and close to $0$ otherwise. In this case we shall have:
\begin{equation}
\F \Rad^* v(\xi) = \frac{1}{\sqpie^{n}} \F\Phi\left(\frac{\norm[2]{\xi}}{\Omega}\right).
\label{eq:filter_factor1}
\end{equation}

This is only achievable with $\Rad^* v$ being radially symmetrical. By looking for $v$ even and taking Theorem~\ref{theo:slice_adjoint}, this can be guaranteed by having $v$ (and, hence $\F v$) independent of its first variable. By abuse of notation, we have $v(\Theta, s) = v(s)$. Now Theorem~\ref{theo:slice_adjoint} says:
\begin{equation}
\F\Rad^* v(\xi) = 2 \cdot \big(\sqrt{2\pi}\, \big)^{n-1} \norm[2]{\xi}^{-n+1} \F v(\norm[2]{\xi}).
\label{eq:adjoint_filter_factor}
\end{equation}
From~\eqref{eq:filter_factor1} and~\eqref{eq:adjoint_filter_factor} we conclude that:
\begin{equation}
\F v(\sigma) = \frac{1}{2} \cdot \frac{1}{\sqpie^{2n-1}} |\sigma|^{n-1}\F\Phi\left(\frac{\sigma}{\Omega}\right).
\end{equation}

Multiple filter factors have been previously suggested, only one of them is presented here for $n=2$. For the ideal low pass filter, i.e.\ $\F\Phi(\sigma)=1$ for $|\sigma|\leq 1$, otherwise $0$, we obtain the \emph{Ram-Lak filter}, suggested first by \cite{RL71}:
\[ v_{\text{Ram-Lak}}(s) = \frac{\Omega^2}{4\pi^2} \left(\sinc(s)-\frac{1}{2} \sinc^2 \left(\frac{s}{2}\right)\right). \]
For the discretised version we compute $v_{\text{Ram-Lak}}$ for $s=\frac{\pi}{\Omega}l, l\in\mathbb{Z}$:
\begin{equation*}
v_{\text{Ram-Lak}}\left(\frac{\pi}{\Omega}l\right) = \frac{\Omega^2}{2\pi^2}
\left\{
\begin{array}{rl}
1/4, & l=0,\\
0, & l \neq 0 \text{ even},\\
-1/(\pi^2l^2), & l \text{ odd}.
\end{array}
\right.
\end{equation*}

\section{Discretisation of the system} \label{sec:discretisation}

\subsection{Parallel scanning geometry} \label{subsec:parallel_scanning_geometry}
The exact choice of known line integrals bears a significance in the derivation of the discretised reconstruction.
In this subsection we present the simplest and most straightforward scanning geometry, the parallel scanning geometry (see Fig.~\ref{fig:sample_geometries:b}). This modality is designed for two dimensional slices ($n=2$) and involves equidistant parallel beams under angles evenly distributed throughout the entire angular range. Therefore, the Radon-transform $\Rad f$ is available for
\begin{align*}
\{ (\Theta_j, s_l) \,\big|\,& \Theta_j = (\cos \varphi_j, \sin\varphi_j)^T, \varphi_j = j\,\Delta\varphi, j=\overline{0,p-1};\\
& s_l = l\,\Delta s , l=\overline{-q, q} \}.
\end{align*}
Here $\Delta\varphi=\pi/p$ is the equal angular step and hence $0 \leq \varphi_j < \pi$. Traversing the half circle is enough because of the even property of the Radon-transform. Besides, recall the inversion formula~\eqref{eq:continuous_inverse_half} for half balls. 

The detector spacing $\Delta s$ needs to be chosen such that the body under examination is covered. Hence, we assume that the hypothesis model $f \in \Sch(\R^2)$ vanishes outside a reconstruction circle with radius $\varrho$. In that case let $\Delta s := \varrho/q$. Note that in this case $\Rad_\Theta f(s)$ also vanishes if $|s| > \varrho$.

\subsection{Discretising the Radon-transformation} \label{subsec:disc_Radon}

During the analysis of different reconstructions we made extensive use of a simple discretisation of the Radon-transformation. The advantage of having an artificial way of producing the forwardprojection is the ability to reuse existing reconstructions in simulations and in designing and evaluating own reconstruction methods. Furthermore, algebraic iterative reconstruction methods apply the forwardprojection during their computations.

In our case, the attenuation coefficient function $f \in \Sch(\R^2)$ (in CT almost always 2D slicing is used) is assumed to vanish outside a unit ball with radius $\varrho$. Therefore $f$ is discretised into an equidistant grid along all dimensions and is represented by a two dimensional, finite, discrete array $f^D \in \R^{\mathcal{X}\times \mathcal{Y}}$, where
\[ \mathcal{X} = \mathcal{Y} = \{-q \,\Delta s,\, (-q+1)\,\Delta s, \ldots, \, (q-1)\,\Delta s,\, q\,\Delta s \, (, (q+1)\,\Delta s)\} \text{ and}\]
\[ f^D(k \, \Delta s, l \, \Delta s)=f(k \, \Delta s, l \, \Delta s), \quad \Delta s = \frac{\varrho}{q}. \]

Denote by $\mathscr{R}_{\varrho} \defeq \{ f^D \in \R^{\mathcal{X} \times \mathcal{Y}} \,|\, f\in \Sch(\R^2);\, f(x)=0, \text{if } \norm[2]{x} > \varrho \}$ the set of all such arrays.

The discrete Radon-transformation has, then, a number of viable implementations: line model (integral computed based on the nearest neighbour model, i.e.\ integral becomes a weighted sum of traversed pixels, weighted by length of line within pixels), strip model (rays and detector cells have a non-zero width and line segments within pixels become strips) or it is possible to conduct physical simulations. 

We choose to evaluate the integral based on an equidistant set of points along lines with values being interpolated bilinearly:
\[ \Rad f(\Theta_j, s_l) \approx \Rad^D f^D(\Theta_j, s_l) \defeq \sum_k \BilInterp[f^D](s_l\Theta_j + k \, \Delta s \, \Theta_j^\perp). \]
In general an interpolation method is defined as follows: let $\chi \in L^1(\R^2) \cap L^2(\R^2)$ be a function, called a \emph{reconstruction kernel}. The interpolation with kernel $\chi$ is the linear operator, for which, in case of $f \in\R^{\mathcal{X}\times \mathcal{Y}}$, we have:
\begin{equation*}
\mathscr{I}_\chi f^D = \chi \ast \sum_{(k,l) \in \mathcal{X}\times\mathcal{Y}} f^D(k \, \Delta s, l \, \Delta s) \delta_{k\,\Delta s, l \,\Delta s}. 
\end{equation*}
Here $ \delta_{k\,\Delta s, l \,\Delta s}$ denotes the Dirac delta translated to $(k\,\Delta s, l \,\Delta s)$. In case of the bilinear interpolation we have that $\chi (u,v) = \Lambda(u, v) =  |1-u/\Delta s| \cdot |1-v/\Delta s|$, if $u,v \in [-\Delta s, \Delta s]$ and zero otherwise. Thus,
\[ \Rad^D f^D(\Theta_j, s_l) = \sum_k \mathscr{I}_\Lambda f^D (s_l\Theta_j + k \, \Delta s \, \Theta_j^\perp). \]

The two subsections~\ref{subsec:parallel_scanning_geometry} and \ref{subsec:disc_Radon} imply that in the fully (both in direct and Radon domain) discretised model, functions defined on $\R^{N\times N}$, $N=2q+1$, are mapped to $\R^{p \times (2q+1)}$ by a linear operator. By denoting this finite dimensional operator with $\Rad$ as well, our task has become to somehow invert $\Rad \in Lin(\R^{N\times N}, \R^{p \times (2q+1)})$. This task is impeded by a possibly low amount of measurement (which turns out to be desirable), hence a large dimensional kernel space $\ker\Rad$ and by the presence of noises in measurements, which could make the equations inconsistent. What is more, the discretised Radon-transformation inverse problems tend to be ill-conditioned, i.e.\ the ratio between the largest and lowest singular values of $\Rad$ is large and pseudo-inverting may admit reconstructions that magnify measurement noises.

For further, more detailed discussion, the reader is referred to \cite{KS01} and the lecture notes \cite{BGJ20}.

\subsection{Convolutional properties of the discrete Radon transformation} \label{subsec:discrete_LI}

Recall Theorem~\ref{theo:continuous_LI}, which stated that $\Rad^*\Rad$ in the continuous case is a linear, shift invariant operator. A very similar result could be achieved in the discrete settings, formally stated by the followings. All notations correspond to the formalism introduced in the previous two subsections. 

This derivation was done by us to facilitate the transition between continuous inversion models and discretised algorithms.

First, the continuous adjoint shown in Theorem~\ref{theo:adjoints} is replaced by an arbitrary quadrature rule:
\[ \mathcal{B} \defeq \left( h \rightarrow \sum_{j=0}^{p-1} \alpha_j h(\Theta_j, \langle \cdot, \Theta_j\rangle)  \right). \]

We also define the discrete convolution, as 
\[ (v \circledast g)(\Theta_j, s) = \Delta s \sum_{l=-q}^{q} v(s-s_l)g(\Theta_j, s_l). \]

\begin{theorem}
\label{theo:discrete_LI}
Let $q$ be chosen such that $\Delta s = \varrho / q \leq \pi / \Omega$. Let $\chi$ be an $\Omega$-band-limited reconstruction kernel, $v$ an $\Omega$-band-limited filter. In this case the operator:
\[ \mathcal{W}: f^D \rightarrow \mathcal{B}\{v \circledast \Rad^D f^D \} \]
is linear and shift invariant on $\mathscr{R}_{\varrho}$, as long as shifting keeps the function in $\mathscr{R}_{\varrho}$ (i.e.\ the shifting of the array only discards zeros).
\end{theorem}

\begin{remark}
This would actually mean, that under some mild conditions the discrete $\Rad^*\Rad$ is also a linear, shift-invariant system,  because we could choose $v$ to be almost the Dirac-delta while still being band-limited. Also, with $v$ being chosen a band-limited ramp-filter, the discrete $\Rad^+\Rad$ becomes an LI-operator.
\TODOm{picture}
\end{remark}

\begin{proof}[Proof of Theorem~\ref{theo:discrete_LI}]
Due to its definition, $\Rad^D$ is linear. The discrete convolution and the quadrature rule are also linear, hence the linearity of $\mathcal{W}$. 

Let us denote $\mathcal{L}_{x_0}h: x \rightarrow h(x-x_0)$ the translation with $x_0$. The case of $h=f^D \in\mathscr{R}_{\varrho}$ is only meaningful, if  $x_0=(a\,\Delta s, b\,\Delta s)$. Assume, that $\mathcal{L}_{x_0} f^D$ is still within the $\varrho$-radius ball.

The interpolation $\mathscr{I}_{\chi}$ is obviously shift-invariant. Also, since $\chi$ is $\Omega$-band-limited, $\mathscr{I}_{\chi} f^D$ becomes $\Omega$-band-limited. This, in conjunction with the fact that $\Delta s \leq \pi/\Omega$, implies
\begin{align}
\Rad^D f^D(\Theta_j, s_l) &= \sum_k \mathscr{I}_\chi f^D  (s_l\Theta_j + k \, \Delta s \, \Theta_j^\perp) = \nonumber \\
&= \int_{\Span\{\Theta_j\}^\perp} \mathscr{I}_\chi f^D(s_l \Theta_j + y)dy = \Rad \mathscr{I}_\chi f^D (\Theta_j, s_l). \label{eq:RD_FD}
\end{align}
Due to the slice theorem~\ref{theo:slice}, $\Rad\mathscr{I}_\chi f^D(\Theta_j, \cdot)$ is also $\Omega$-band-limited. With $v$ also $\Omega$-band-limited, applying~\eqref{eq:RD_FD} we arrive to the conclusion that
\begin{align} 
(v \circledast \Rad^D f^D)(\Theta_j, s) &= \Delta s \sum_{l=-q}^{q} v(s-s_l)\Rad\mathscr{I}_\chi f^D(\Theta_j, s_l) = \nonumber \\
&= \int_{-q}^{q} v(s-t)\Rad \mathscr{I}_\chi f^D(\Theta_j, t)dt = (v \ast \Rad \mathscr{I}_\chi f^D)(\Theta_j, s). \label{eq:v_RD_FD}
\end{align}

Similarly, $\mathscr{I}_{\chi} \mathcal{L}_{x_0} f^D$ is $\Omega$-band-limited, and, therefore applying~\eqref{eq:v_RD_FD} for $\mathcal{L}_{x_0} f^D$ gives us
\[ (v \circledast \Rad^D \mathcal{L}_{x_0} f^D)(\Theta_j, s) = (v \ast \Rad \mathscr{I}_\chi \mathcal{L}_{x_0} f^D)(\Theta_j, s) = (v \ast \Rad \mathcal{L}_{x_0} \mathscr{I}_\chi f^D)(\Theta_j, s). \]

It is straightforward to derive and visualise that for $f \in \Sch(\R^2)$, we obtain
\[ \Rad \mathcal{L}_{x_0} f(\Theta, s) = \Rad f(\Theta, s - \langle x_0, \Theta \rangle) = \mathcal{L}_{\langle x_0, \Theta \rangle} \Rad f (\Theta, s). \]
This, however, together with the shift-invariance of convolution means that 
\[ (v \circledast \Rad^D \mathcal{L}_{x_0}f^D)(\Theta_j, s) = \mathcal{L}_{\langle x_0, \Theta_j \rangle} \{v \ast \Rad \mathscr{I}_\chi f^D\}(\Theta_j, s) = \mathcal{L}_{\langle x_0, \Theta_j \rangle} \{ v \circledast \Rad^D f^D\} (\Theta_j, s). \]
This way,
\begin{align*}
\mathcal{W}\mathcal{L}_{x_0} f^D(x) &= \sum_{j=0}^{p-1} \alpha_j (v \circledast \Rad^D \mathcal{L}_{x_0}f^D)(\Theta_j, \langle x, \Theta_j\rangle)= \\
&= \sum_{j=0}^{p-1} \alpha_j \mathcal{L}_{\langle x_0, \Theta_j \rangle}(v \circledast \Rad^D f^D)(\Theta_j, \langle x, \Theta_j\rangle) =\\
&= \sum_{j=0}^{p-1} \alpha_j (v \circledast \Rad^D f^D)(\Theta_j, \langle x - x_0, \Theta_j\rangle) = \mathcal{L}_{x_0}\mathcal{W} f^D(x). \qedhere
\end{align*}

\end{proof}

\mysection{Reconstruction via discretised FBP for standard parallel geometry}{FBP for standard parallel geometry}
\label{sec:discretised_FBP}

Before, in subsection~\ref{subsec:continuous_FBP}, we presented an exact inversion formula for the Radon-operator. It was also called the continuous filtered backprojection (FBP). Nevertheless, in practical cases the number of projections ($\Rad_\Theta$) is finite. Besides, a single projection is limited spatially and, likewise, discretised in a finite number of sampling points based on the detector cells' arrangement in the detector panel. Furthermore, the reconstruction may only take place in a finite number of points (practically, a grid). Consequently, continuous inversion formulas need to be discretised in a way to fit the available measurements. We will see, the discretised FBP is going to be able to reconstruct exactly the original hypothesis model, given by the continuous inversion formula, on a grid, under conditions due to sampling theoretical considerations.

We present the discretisation of the FBP for the settings of subsection~\ref{subsec:parallel_scanning_geometry} without giving the proofs.
Let $v$ be $\Omega$-band-limited filter and let $g=\Rad f$ (also $\Omega$-band-limited).
We further assume that the function $f$ is $\Omega$-band-limited. This case also implies $\Rad_\Theta f$ being $\Omega$-band-limited due to the Fourier slice theorem~\ref{theo:slice}.

\begin{lemma}
\label{lemma:convolution_trapezoid}
Assume that $\Delta s \leq \pi/\Omega$ holds. This is equivalent to having $q \geq \Omega\varrho / \pi$. In that case the filtering is possible with the trapezoidal rule:
\[ (v \ast g)(\Theta, s) = \Delta s \sum_{l=-q}^{q} v(s-s_l)g(\Theta, s_l). \]
\end{lemma}

\begin{lemma}
\label{lemma:backprojection_trapezoid}
Assume that $\Delta\varphi \leq \pi/(\Omega\varrho)$, which is equivalent to $p \geq \Omega\varrho$. Then, the backprojection may be computed via the trapezoidal rule:
\begin{equation*}
\int_{S^{1}} (v \ast g)(\Theta, \langle x, \Theta \rangle) d\Theta = \frac{2\pi}{p} \sum_{j=0}^{p-1} (v \ast g)(\Theta_j, \langle x, \Theta_j \rangle).
\end{equation*}
\end{lemma}

Combining Lemmas~\ref{lemma:convolution_trapezoid} and \ref{lemma:backprojection_trapezoid} yields two different approaches. The first one, the more direct one leads to:
\begin{equation*}
f(x) = (\Rad^*v \ast f)(x) = \frac{2\pi}{p} \Delta s\sum_{j=0}^{p-1}  \sum_{l=-q}^{q} v(\langle x, \Theta_j \rangle-s_l)g(\Theta_j, s_l).
\end{equation*}
This formulation is computationally more expensive than the second approach, which involves precomputing $(v\ast g)(\Theta_j, s_k)$ for every $j$ and $k$ and, afterwards, computing $(v \ast g)(\Theta_j, \langle x, \Theta_j \rangle)$ via linear interpolation. Since $\Delta s \leq \pi/\Omega$ the values $(v\ast g)(\Theta_j, s_k)$ uniquely determine $v\ast g$ (though, not by linear interpolation). In this case the scheme is as follows. Precompute:
\[ h_{j,k} = \Delta s \sum_{l=-q}^{q} v(s_k-s_l)g(\Theta_j, s_l), \quad k=\overline{-q, q}, \quad j=\overline{0,p-1}. \]
Then for an arbitrary $x$:
\[ f(x)\approx \frac{2\pi}{p} \sum_{j=0}^{p-1} ((1-\nu)h_{j,k} + \nu h_{j,k+1}), \]
where for $t=\langle x, \Theta_j \rangle / \Delta s$ we have $k = \lfloor t \rfloor$ and $\nu = t-k$.

\begin{remark}
This reconstruction method builds upon a finite number of projections and in undersampled cases it certainly won't be precise. Nonetheless, it is provable that even in the case of less projections than sampling-theoretically needed, the method implements a discrete pseudo-inverse of $\Rad^D$, that is $(\Rad^D)^+$. In subsection~\ref{subsec:discrete_LI} we already made a remark for Theorem~\ref{theo:discrete_LI} that the discrete $\Rad^+\Rad$ is also an LI operator.
\end{remark}

\section{Algebraic Reconstruction Techniques (ART)} \label{sec:art}

We have pointed out multiple times that the Radon-transformation $\Rad$ in both the continuous and discrete settings becomes a linear operator and the reconstruction problem is equivalent to solving a linear equation. Due to the difficulties enlisted in subsection~\ref{subsec:disc_Radon}, traditional linear equation solving methods (Gauss-elimination) are not suited for these problems. Another calamity with these methods would be the unavailability of the otherwise enormous matrix representation of $\Rad$ even though being sparse. Reconstruction should be achieved through the functional usage of the operators $\Rad,\Rad^*,\Rad^+$. Therefore, numerous iterative methods have been suggested for solving general linear equations.

\subsection{Kaczmarz- and Cimmino-iterations}

The Kaczmarz-iteration \cite{Ka37} takes separate linear equations and iteratively sweeps through each of them and projects the current estimate onto the affine subspace described by the current equation. \TODOm{insert picture} Formally this is defined as having a system of $p$ linear equations 
\begin{equation*}
\Rad_j f=g_j\text{, where } j=\overline{1, p}
\end{equation*}
and $\Rad_j:H \rightarrow H_j$ are surjective, bounded linear operators between Hilbert spaces $H$ and $H_j$.
Our goal is to project a current solution estimate $f$ onto the affine subspace $\Rad_j f = g_j$ (which is a translation of the kernel space $\ker R_j$). This is done by the following operator:
\begin{equation}
P_j f \defeq f + \Rad_j^+ (g_j - \Rad_j f) = f + \Rad_j^*(\Rad_j\Rad_j^*)^{-1} (g_j - \Rad_j f).
\label{eq:kaczmarz_projection}
\end{equation}
Indeed, $\Rad P_j f = \Rad f + g_j - \Rad_j f = g_j$, thus $P_j$ maps to the affine subspace. Moreover, for $\tilde{f} \in \ker R_j$
\[ \langle P_j f - f, \tilde{f} \rangle = \langle \Rad^+ (g_j - \Rad_j f), \tilde{f} \rangle = 0, \]
because $\Rad^+$ maps to $(\ker \Rad)^\perp$, hence the orthogonality of $P_j$.	

The Kaczmarz-iteration is defined as:
\begin{equation}
\begin{aligned}
f^{(0)} &\text{ is an original estimate},\\
f^{(k,0)} & \defeq f^{(k)},\\
f^{(k, j)} &\defeq P_j f^{(k, j-1)}, \text{ for }j=\overline{1,p},\\
f^{(k+1)} &\defeq f^{(k,p)}.
\end{aligned}
\label{it:kaczmarz}
\end{equation}

Most of the time, because of convergence guarantees, the projection step~\eqref{eq:kaczmarz_projection} is relaxed to
\[ P_j^{(\omega)} f \defeq f + \omega \Rad_j^*(\Rad_j\Rad_j^*)^{-1} (g_j - \Rad_j f),\text{ i.e. } P_j^{(\omega)} = (1-\omega)\I_H + \omega P_j. \]

The most general form is obtained by not only introducing relaxation but also replacing $\Rad_j \Rad_j^*$ with a $C_j \succeq \Rad_j \Rad_j^*$ positive definite, symmetric operator:
\begin{equation}
\begin{aligned}
f^{(k, j)} &\defeq f^{(k,j-1)} + \omega \Rad_j^*C_j^{-1} (g_j - \Rad_j f^{(k,j-1)}), \text{ for }j=\overline{1,p}, \\
f^{(k,0)} & \defeq f^(k), f^{(k+1)} \defeq f^{(k,p)}. 
\end{aligned}
\label{it:kaczmarz_relaxed_C}
\end{equation}

It is obvious that the general Kaczmarz-method~\eqref{it:kaczmarz_relaxed_C} is a sequential model splitting up the entire measurement operator $\Rad = (\Rad_1, \Rad_2, \ldots, \Rad_p)^T$ into its components. The choice of block is arbitrary. For instance, in the original idea behind the Kaczmarz-method, all $\Rad_j$ consisted of one single row of a matrix, i.e.\ all $\Rad_j$ had rank $1$. This could be substituted by larger rank blocks. One important example is the full size iteration, that is $p=1$ and each larger iteration step comprises a single projection step:
\begin{equation}
f^{(k+1)} \defeq f^{(k)} + \omega\Rad^*C^{-1} \big(g - \Rad f^{(k)}\big), \text{ where } C \succeq \Rad\Rad^*.
\label{eq:kaczmarz_1_block}
\end{equation}
A special case is the Landweber-iteration, which we are going to analyse more deeply and use during our research work. In that instance $C = \gamma\I_H$ and $\tilde{\omega} \defeq \omega / \gamma$ and 
\[ f^{(k+1)} \defeq f^{(k)} + \tilde{\omega}\Rad^*\big(g - \Rad f^{(k)}\big). \]

Another possible method is the Cimmino-iteration (also called simultaneous iterative reconstruction technique, SIRT, see \cite{TL90}), which, contrary to the Kaczmarz-method, is a fully parallel method and consists of merging all inner steps of~\eqref{it:kaczmarz_relaxed_C} into a single update step:
\[ f^{(k+1)} = f^{(k)} + \omega \sum_{j=1}^{p} \Rad_j^* C_j^{-1} \big(g_j - \Rad_j f^{(k)}\big). \]
Another entire category of iterative approaches include the simultaneous algebraic reconstruction technique (SART, \cite{AK84}) with proven convergence guarantees (see \cite{JW03}).

Let us state here a theorem about the convergence of the general Kaczmarz-iteration\ \eqref{it:kaczmarz_relaxed_C}. We denote $\Rad = (\Rad_1, \Rad_2, \dots, \Rad_p)^T$ and $g = (g_1, g_2, \ldots, g_p)^T$. Firstly, it is to be observed that the iteration is not able to touch the projection of the initial guess $f^{(0)}$ onto the kernel space $\ker\Rad$, because the update change always maps, due to $\Rad_j^*$, to $\supp\Rad_j \defeq (\ker\Rad_j)^\perp \subseteq (\ker\Rad)^\perp$.

\begin{theorem}
\label{theo:kaczmarz_convergence}
Assume that the equation $\Rad f = g$ is consistent and $\omega \in (0,2)$. Then~\eqref{it:kaczmarz_relaxed_C} is convergent and
\[ \lim_{k \rightarrow \infty} f^{(k)} = \Rad^+ g + \proj_{\ker\Rad} f^{(0)}. \]
The same result holds if the equation is not necessarily consistent, but the Hilbert-space $H$ is finite dimensional.
\end{theorem}
Therefore, the iteration does its best to compute the pseudo-inverse except for touching the kernel space projection. The proof is to be outlined only for the Landweber-iteration ($C=\gamma\I$) and only in finite dimension, because we wish to analyse the exact change of error during iterations.

\subsection{Landweber-iteration. Variational viewpoint} \label{subsec:Landweber}
As already presented, the Landweber-iteration is as follows:
\begin{equation}
f^{(k+1)} \defeq f^{(k)} + \omega\Rad^*\big(g-\Rad f^{(k)}\big).
\label{it:Landweber}
\end{equation}

Note that the Landweber-iteration step is in fact the gradient descent iteration step for the least squares problem
\begin{align}
\min_{f} & \quad\mathcal{J}(f) \defeq \frac{1}{2} \big\lVert\Rad f - g \big\rVert_2^2.
\end{align}
Indeed:
\[ \nabla \mathcal{J}(f) = (\Rad f - g)^*\Rad, \]
and, hence, the GD-step is:
\[ f^{(k+1)} = f^{(k)} - \omega \nabla^* \mathcal{J}\big(f^{(k)}\big) = f^{(k)} + \omega\Rad^*\big(g-\Rad f^{(k)}\big).\]
This sort of reinterpretation present so often is the main reason of the preference for iterative methods over the FBP. The variational interpretation allows the direct injection of penalisation terms promoting prior information.

We now prove the Theorem~\ref{theo:kaczmarz_convergence} for the Landweber-iteration, when $H$ is finite dimensional. Note that in~\eqref{it:Landweber} the $\gamma$ multiplier was not explicitly noted, hence the convergence criteria for $\omega$ in this form will change.

First, rewrite~\eqref{it:Landweber}:
\begin{equation}
f^{(k+1)} = (\I - \omega\Rad^*\Rad) f^{(k)} + \omega\Rad^*g.
\label{it:Landweber_2}
\end{equation}
Therefore, by denoting $Q \defeq \I - \omega\Rad^*\Rad$ and tracing it back to the $0^{\text{th}}$ element we obtain:
\begin{equation}
f^{(k)} = Q^k f^{(0)} + \sum_{j=0}^{k-1} Q^j \omega\Rad^*g.
\label{it:Landweber_general}
\end{equation}
From now an SVD-based analysis is introduced. Let $\sigma_i$ be the $i^{\text{th}}$ largest singular value of $\Rad$ with right singular vectors $v_i$ and left singular vectors $u_i$ for $1 \leq i \leq rank(\Rad) \eqdef r$. Let us use the Dirac notations: $|x \rangle : \R \rightarrow H, |x \rangle \lambda = \lambda x$ and $\langle x | : H \rightarrow \R, \langle x |v = \langle x,v \rangle$. With this notation, for example, $|v_i\rangle \langle v_i|$ is the orthogonal projection onto $\Span\{v_i\}$ and, thus, $\I = \sum_{i=1}^r |v_i \rangle \langle v_i| + \proj_{\ker\Rad}$. In this case we have 
\begin{equation*}
\Rad = \sum_{i=1}^r \sigma_i | u_i \rangle \langle v_i|, \qquad 
\Rad^* = \sum_{i=1}^r \sigma_i |v_i \rangle \langle u_i|, \qquad
\Rad^+ = \sum_{i=1}^{r} \frac{|v_i \rangle \langle u_i|}{\sigma_i},
\end{equation*}
and also
\begin{equation*}
\Rad^*\Rad = \sum_{i=1}^r \sigma_i^2 | v_i \rangle \langle v_i|, \qquad
\Rad\Rad^* = \sum_{i=1}^r \sigma_i^2 | u_i \rangle \langle u_i|.
\end{equation*}
Therefore, by rewriting $Q$, the following is obtained:
\begin{align*}
Q &= \I - \omega \Rad^*\Rad = \I - \omega \sum_{i=1}^r \sigma_i^2 | v_i \rangle\langle v_i| = \\
&= \sum_{i=1}^r \big(1-\omega \sigma_i^2\big)|v_i \rangle \langle v_i| + \proj_{\ker\Rad}.
\end{align*}
Because $v_1, v_2, \ldots, v_r$ span out $\supp\Rad$, the powering of $Q$ becomes simply
\[ Q^j = \sum_{i=1}^r \big(1-\omega \sigma_i^2\big)^j|v_i \rangle \langle v_i| + \proj_{\ker\Rad}. \]
Our goal is the choice of $\omega$ in order to guarantee the convergence of iteration~\eqref{it:Landweber_general} under any hypothesis $f^{(0)}$. This translates to convergence even if $f^{(0)} = 0$. This implies that the second term $\sum_{j=0}^{k-1} Q^j \omega\Rad^*g$, dependent only on $\omega$ and not $f^{(0)}$, is convergent on its own. Hence, the first term $Q^k f^{(0)}$ should be convergent on its own with any choice of $f^{(0)}$. However,
\begin{equation}
Q^k f^{(0)} = \sum_{i=1}^r \big(1-\omega\sigma_i^2\big)^k \langle v_i, f^{(0)} \rangle v_i + \proj_{\ker\Rad}f^{(0)}.
\label{it:Landweber1}
\end{equation}
With suitable choices of $f^{(0)}$ any term $\langle v_i, f^{(0)} \rangle$ could be non-zero, hence for convergence the following necessary condition is derived:
\begin{equation} 
|1-\omega\sigma_i^2| < 1, \text{ for any }i\in\overline{1,r}.
\label{eq:condition_Landweber0}
\end{equation}
For this $\omega$ is positive and from $1-\omega\sigma_i^2 > -1$ we get that $\omega <2/\sigma_i^2$, for any $i$. Hence, the necessary condition of convergence is:
\begin{equation}
0 < \omega < \frac{2}{\sigma_{\text{max}}^2} = \frac{2}{\norm[2]{\Rad^*\Rad}}.
\label{eq:condition_Landweber}
\end{equation}
Note that even in this case, the iteration $\eqref{it:Landweber1}$ converges to $\proj_{\ker\Rad} f^{(0)}$, which is consistent with the earlier observation that the iteration cannot touch the (most likely wrong) kernel space component of the hypothesis. 

For the sufficiency of this condition and the error analysis we now elaborate more on the second term. Assume $f^{(0)}=0$, and note that $\Rad^*g \in \supp\Rad$. Thus,
\begin{equation*}
f^{(k)} = \sum_{j=0}^{k-1} Q^j\omega\Rad^*g = \sum_{j=0}^{k-1} \sum_{i=1}^{r} \omega \big(1-\omega\sigma_i^2\big)^j |v_i \rangle \langle v_i|\Rad^* g
\end{equation*}
Using the fact, that $|v_i \rangle \langle v_i|\Rad^*=\sigma_i |v_i \rangle \langle u_i|$, we get:
\begin{align*}
f^{(k)} &= \sum_{i=1}^{r} \sum_{j=0}^{k-1} \omega\sigma_i \big(1-\omega\sigma_i^2\big)^j \langle u_i, g \rangle v_i =\\ 
&= \sum_{i=1}^{r} \omega\sigma_i\cdot\frac{1-\big(1-\omega\sigma_i^2\big)^k}{1-\big(1-\omega\sigma_i^2\big)}\langle u_i, g \rangle v_i =\\
&= \sum_{i=1}^{r} \Big(1-\big(1-\omega\sigma_i^2\big)^k\Big)\frac{\langle u_i, g \rangle}{\sigma_i} v_i.
\end{align*}
Compare this to the equation $\Rad^+ g = \sum_{i=1}^{r} \frac{\langle u_i, g \rangle}{\sigma_i} v_i$. Denoting $\phi_i^{(k)} = 1-\big(1-\omega\sigma_i^2\big)^k$ we find that for $0 < \omega < 2/\sigma_{\text{max}}^2$ we have $\phi_i^{(k)} \rightarrow 1, \text{ as } k \rightarrow \infty$. 
Hence, \eqref{eq:condition_Landweber} is a necessary and sufficient condition of convergence of the iteration~\eqref{it:Landweber_general}.
Besides that,
\begin{equation*}
	\phi_i^{(k)} = 1-\big(1-\omega\sigma_i^2\big)^k \approx \left\{
	\begin{array}{rl}
		1, & \text{if } \sigma_i \gg 1/\sqrt{k\omega}, \\
		k\omega\sigma_i^2, & \text{if } \sigma_i \ll 1/\sqrt{k\omega}.
	\end{array}
	\right.
\end{equation*}
Thus, SVD components belonging to larger singular values become more saturated ($\phi_i^{(k)} \approx 1$) within fewer number of iterations and the choice of a final iteration number $k$ becomes a regularisation parameter.

Let $\bar{f}$ be the exact solution, $\bar{g} \defeq \Rad\bar{f}$, $e_g = \bar{g}-g$. Let $\bar{f}^{(k)}$ be the iteration vector after applying $k$ Landweber-iterations to the noise-free measurement data $\bar{g}$ (also assuming $\bar{f}^{(0)}=0$). The overall error after $k$ iterations is split between two components, the noise-independent reconstruction error and the noise error between the actual and the noise-free reconstruction:
\[ \bar{f} - f^{(k)} = \big(\bar{f} - \bar{f}^{(k)}\big) + \big(\bar{f}^{(k)} - f^{(k)}\big). \]
On one hand, having $\langle u_i, \bar{g} \rangle = \sigma_i \langle v_i, \bar{f} \rangle$, we derive that
\begin{align*}
\bar{f} - \bar{f}^{(k)} &= \sum_{i=1}^r \langle v_i, \bar{f} \rangle v_i + \proj_{\ker\Rad} \bar{f} - \sum_{i=1}^{r} \phi_i^{(k)}\langle v_i, \bar{f} \rangle v_i\\
&= \proj_{\ker\Rad} \bar{f}  + \sum_{i=1}^r \big(1-\phi_i^{(k)}\big) \langle v_i, \bar{f} \rangle v_i.
\end{align*}
As $k\rightarrow \infty$, the second term approaches to zero and all what is left from the noise-independent error is the kernel space component of the ideal reconstruction. What is more, the noise-independent reconstruction error converges to its minimum decreasingly in norm, since all components are orthogonal and $1-\phi_i^{(k)}$ is decreasing to $0$.

On the other hand,
\begin{equation*}
\bar{f}^{(k)} - f^{(k)} = \sum_{i=1}^r \phi_i^{(k)} \frac{\langle u_i, \bar{g}-g \rangle}{\sigma_i}v_i
=\sum_{i=1}^r \phi_i^{(k)} \frac{\langle u_i, e_g \rangle}{\sigma_i}v_i.
\end{equation*}

To sum up, while the noise-independent reconstruction error converges to its minimum decreasingly, the noise error increases at all components, where the ratio $\frac{\langle u_i, e_g\rangle}{\sigma_i}$ is large. In practice this usually means that the overall error decreases until a certain iteration and, then, the noise error takes over and overall error begins increasing. This empirical effect is called \emph{semi-convergence}.

\TODOm{introduce picture about semi-convergence}

\chapter[CNNs for inverse problems]{Convolutional neural networks for inverse problems} \label{chap:cnn}

Machine and deep learning has become the widespread, state-of-the-art solution for classification problems and experimentation for applying them to solve linear inverse problems has recently commenced.
In this chapter we are going to shortly relate to the mechanisms of artificial neural networks by constraining ourselves to convolutional nets. After this we aim to present previous successful attempts of adapting the neural paradigm for CT reconstruction purposes.

\section{Artificial Neural Networks}

The general concept behind NNs is to learn conditional distributions via hidden representations. For example, for classification tasks it is rarely conceivable to design a classical algorithm because of our inability to observe useful representations. These conditional distributions could also be well-defined mappings. For example, in the CT reconstruction problem the goal usually is to perform denoising, i.e.\ learn a conditional distribution of the ideal, noise-free reconstruction with respect to the noisy reconstruction.

\subsection{Generalisation error and ERM}
In general, if $(X, Y)$ are random variables from the joint distribution $\mathcal{D}$, then our goal is to model as precisely as possible the conditional distribution $\prob_{Y|X}(y | x), (x,y) \in \supp D$. Whenever the connection is a direct mapping $f:X\rightarrow Y$, the conditional distribution to be learned becomes $\prob_{Y|X}(y|x) = \delta_{f(x)}(y)$. In many applications, including ours, the modelling is deterministic or, put otherwise, functional. In such cases, the set of possible functions that comprise our search space is called \emph{hypothesis class}. Almost always such a hypothesis class, now denoted by $\mathcal{H}$, is given in function of a parameter $\Theta \in \R^d$. The learning problem and the \emph{empirical risk minimisation (ERM)} problem could be derived from two directions giving two different explanations.

On one hand, one usually wishes to minimise a certain metric between the ideal, expected output and the output of the hypothesis. For our task we assume this metric to be the squared euclidean distance. Therefore, the generalised risk minimisation takes the following form:
\begin{IEEEeqnarray*}{r'l}
\min_{h \in \mathcal{H}} & \meanlim_{(X,Y)\sim D} \Big[ \frac{1}{2} \big\lVert Y - h(X) \big\rVert^2 \Big].
\end{IEEEeqnarray*}
Assuming that we possess a sample vector $S = [(x_i, y_i) \,|\, i=\overline{1,m}\,] \sim \mathcal{D}^m$, the mean could be rewritten to an average. This yields the empirical risk minimisation problem:
\begin{IEEEeqnarray}{r'l}
\min_{h \in \mathcal{H}} & \frac{1}{2m} \sum_{i=1}^m \norm{y_i-h(x_i)}^2.
\label{opt:ERM}
\end{IEEEeqnarray}
Since it is often difficult for the system to retrieve exactly the perfect representations, this objective is often supplemented by a priori information about a possibly reasonable hypothesis. This prior knowledge is expressed as a penalisation term $\Omega$ controlled by a multiplier $\lambda$:
\begin{IEEEeqnarray}{r'l}
\min_{h \in \mathcal{H}} & \frac{1}{2m} \sum_{i=1}^m \norm{y_i-h(x_i)}^2 + \lambda \Omega(h).
\label{opt:ERM_regularised}
\end{IEEEeqnarray}

One the other hand, we could start off from a maximum likelihood (ML) estimation problem:
\begin{IEEEeqnarray}{r'l}
\max_{h \in \mathcal{H}} & \prod_{i=1}^m \prob (y_i | x_i, h).
\label{opt:ML}
\end{IEEEeqnarray}
Moreover, we could model a priori information about a possibly reasonable hypothesis by going for a maximum a posteriori estimation:
$\max_{h \in \mathcal{H}} \prod_{i=1}^m \prob (y_i, h| x_i)$.
Using that $\prob (y, h | x) = \prob (y | x, h) \prob(h)$, we obtain the following form:
\begin{IEEEeqnarray}{r'l}
\max_{h \in \mathcal{H}} & \prod_{i=1}^m \prob (y_i| x_i, h) \prob (h).
\label{opt:ML_regularisation}
\end{IEEEeqnarray}
Now introducing the likelihood terms $l = -\ln \prob$, \eqref{opt:ML} and \eqref{opt:ML_regularisation} take the following forms:
\begin{IEEEeqnarray}{r'l"t"r'l}
\min_{h \in \mathcal{H}} &\frac{1}{m}\sum_{i=1}^{m} l(y_i| x_i, h) & and & \min_{h \in \mathcal{H}} &\frac{1}{m}\sum_{i=1}^{m} l(y_i| x_i, h) + \lambda l(h).
\label{opt:ML_MAP_likelihood}
\end{IEEEeqnarray}
As it is seen, the problems in \eqref{opt:ERM} and \eqref{opt:ERM_regularised} are strongly related to the likelihood formulations of \eqref{opt:ML_MAP_likelihood}. In particular, if the conditional distribution $\prob (y|x,h)$ is in fact a standard Gaussian distribution, then the fidelity term of problems in \eqref{opt:ML_MAP_likelihood} falls back to the L2-norm problems. Furthermore, the regularisation term in \eqref{opt:ERM_regularised}, $\Omega(h)$, besides being a penalisation, incorporates a very specific meaning that a conceivable hypothesis has a certain underlying distribution.

\subsection{Special hypothesis class. Convolutional layers}
One important task of applied machine learning is to choose a proper hypothesis class. For reasons explained more in details in an upcoming section, it is completely reasonable for us to choose convolutional neural networks (CNNs) as our starting point. 
CNNs have convolutional layers as their basic building stone, which apply a discrete convolutional operation, represented by a convolution kernel, to the image inputted and add a supplementary bias. If $l-1$ and $l$ represent the indices of two consecutive layers in a CNN, $z \in \{1, \ldots, Z^{(l)}\}$ is one of the $Z^{(l)}$ channels of layer $l$, $y^{(l-1)}$ is the output of layer $l-1$ with  $Z^{(l-1)}$ channels, $w^{(l)}(\cdot,\cdot,c,z)$ is the convolutional kernel between the $l^{\text{th}}$ layer's $z^{\text{th}}$ channel and the $(l-1)^{\text{th}}$ layer's $c^{\text{th}}$ channel, and $B^{(l)}_{(z)}$ is the bias corresponding to channel $z$, then the output of layer $l$ is the following:
\begin{equation*}
y^{(l)}_{(z)} = \sum_{c=1}^{Z^{(l-1)}} \sum_{(a, b)} y^{(l-1)}_{(c)} (x + a, y + b) \cdot w^{(l)}(a, b, c, z) + B^{(l)}_{(z)}.
\end{equation*}
Convolution kernels ($w^{(l)}(\cdot,\cdot,c,z)$) usually have a relatively small, square-shaped support, most commonly $3\times 3$. This support size is interpreted as the receptive field of a single layer, since it shows how far information from one cell could propagate after one layer.
In deep neural networks, such as the U-Net, a rectified linear function (ReLU) is applied as activation function element-wise on the output of a convolution. The ReLU's definition:
\begin{equation*}
f : \R \to \left[0, \infty\right),
f(x) = \max\{0, x\}
\end{equation*}
CNNs usually profit from non-linear layers, like maximum pooling, which creates a half sized image with values being computed as the maximum of four elements in the input image. This operation has the role of omitting less important information and highlighting the more important features. It reduces the dimensionality of the input space and makes, thus, the output space more condensed even if the input distribution was sparse in the input space.

Besides max pooling layers, obviously convolutional layers could be reducing dimensionality via the usage of strides. The reverse operations of dimension reductive layers consists of two possibilities: the algebraic transpose (adjoint) of dimension reductive convolutions, also called transposed convolutions, or schemes built from upsampling via interpolation followed by a size-keeping convolution.

One further element included by deep fully convolutional neural nets is batch normalisation layer that normalises incoming data batches in order to keep input activations in the well-conditioned region of the activation function and to make it possible to train layers relatively independently.

One important thing related to CNNs is the fact that under some clear definitions they approximate continuous, shift invariant (sometimes called translation equivariant) functions on a compact set arbitrarily well. This result, more outlined in \cite{Zh20, Ya21} resonate with the universal approximation property of dense networks of fully connected layers.

\subsection{U-Net} \label{subsec:UNet}
U-Net was first presented in \cite{RFB15} for the purpose of biomedical image segmentation. The network is one of the first fully convolutional networks consisting of convolutional layers organised into blocks of three. There is a downsampling stream of five stages, made possible by max pooling layers. At each stage, the spatial dimensions of the image size are roughly halved, but at the same time, the number of channels gets doubled. This motivates the network to filter differently abstract object parts. On the upsampling scheme these steps get inverted. After upsampling the per-channel dimensions are doubled, yet the number of channels is cut to its half. Nevertheless, the U-Net uses increasingly many skip connections. Between each stage there is a skip connection and the downstream data is concatenated to the upstream channels. 

The original version proposed in the article did not apply padding and feature maps did not reduce by a factor of two perfectly. For the original purposes, input data was preprocessed and a mirror padding was applied to them. Furthermore, the output of the network was also adjusted for image segmentation specific standards. A slightly modified architecture of the U-Net could be seen in Fig.~\ref{fig:unet_Huang}, here contributed to \cite{HW18}, yet applied in this form by almost the entire community. In the updated architecture, zero padding is applied and the downstream shrinkage is perfect. Upstream expansion, thus, perfectly matches size-wise the corresponding level's block. Hence, the final stage reassures an input size result. The network is finished with a 1x1 kernel convolutional layer to scale the data to the space of the desired output.

\subsection{Optimisation}
As all usual network types, the hypothesis class of convolutional networks is also parameterised by a parameter vector $\mathcal{W} \in \R^D$. This vector contains all convolution kernels, biases and the parameters of affine scaling parameters of the batch normalisation. The ERM problem~\eqref{opt:ERM} becomes:
\begin{IEEEeqnarray}{r'l}
\min_{\mathcal{W} \in \R^D} & \frac{1}{2m} \sum_{i=1}^m \big\lVert y_i - f(x_i, \mathcal{W}) \big\rVert_2^2
\label{opt:ERM_weight}
\end{IEEEeqnarray}
Thus, the hypothesis class becomes $\mathcal{H} = \{ f(\cdot, \mathcal{W}) \,|\, \mathcal{W} \in \R^D \}$. We also extend function $f(\cdot, \mathcal{W})$ to take batch inputs and yield batch output, for $x = (x_1, x_2, \ldots, x_m)$ and $y = (y_1, y_2, \ldots, y_m)$ we have:
\[ f(x, \mathcal{W}) \defeq \big(f(x_1, \mathcal{W}), f(x_2, \mathcal{W}), \ldots, f(x_m, \mathcal{W})\big). \]
In this case the ERM problem~\eqref{opt:ERM_weight} is simplified to:
\begin{IEEEeqnarray}{r'l}
\min_{\mathcal{W} \in \R^D} & \mathcal{L}(x, \mathcal{W}) = \frac{1}{2} \big\lVert y - f(x, \mathcal{W}) \big\rVert_2^2.
\label{opt:ERM_weight_batch}
\end{IEEEeqnarray}
Problem~\eqref{opt:ERM_weight_batch} is solved with different variants of the gradient (steepest) descent method. For the understanding of following chapters we derive here the update step and delta-rule for one inner layer.
We assume that the $l^\text{th}$ layer takes the form $ y^{(l)} = f^{(l)}(x^{(l)}, \mathcal{W}^{(l)}) $. The weight components of everything before and after this layer is fixed for the moment and we may assume that everything before is summarised by an operator $y^{(l-1)} = g^{(l-1)}(x^{(0)})$, everything after, including the loss, is summarised by $G^{(l+1)}(x^{(l+1)})$. The partial update rule in a simple GD for weight component $\mathcal{W}^{(l)}$ reads:
\[ \mathcal{W}^{(l)}_{k+1} = \mathcal{W}^{(l)}_k - \lambda \nabla_{\mathcal{W}^{(l)}}^{*} \mathcal{L}(x^{(0)}, \mathcal{W}_k), \]
where update is done based on the adjoint (or transposed) of the partial gradient.
The function $\mathcal{L}$ is obviously a composition of the three aforementioned components. With a bit of notational abuse with regard to fixed weights, we get:
\[\mathcal{L}(x_0, \mathcal{W}) = G^{(l+1)}\Big( f^{(l)}\big(g^{(l-1)}(x^{(0)}), \mathcal{W}^{(l)}\big) \Big).\]
Therefore, after omitting quite a few inner steps, we obtain:
\begin{align}
\nabla_{\mathcal{W}^{(l)}} \mathcal{L}(x_0, \mathcal{W}_k) &= \big(\nabla G^{(l+1)}\big)(y^{(l)})\cdot \nabla_{\mathcal{W}^{(l)}} f^{(l)}(y^{(l-1)}, \mathcal{W}^{(l)}_k), \nonumber \\
\nabla^*_{\mathcal{W}^{(l)}} \mathcal{L}(x_0, \mathcal{W}_k) &= \nabla^*_{\mathcal{W}^{(l)}} f^{(l)}(y^{(l-1)}, \mathcal{W}^{(l)}_k) \cdot \big(\nabla^* G^{(l+1)}\big)(y^{(l)}).
\label{eq:GD_W}
\end{align}
Formula~\eqref{eq:GD_W} displays a usually huge size multiplication, in the profession usually just called \emph{jacobian-vector product}. The \emph{output gradient} term $\big(\nabla^* G^{(l+1)}\big)(y^{(l)})$ has a dimension equal to the output dimension of the $l^\text{th}$ layer. However, the front, multiplicative, linear term $ \nabla^*_{\mathcal{W}^{(l)}} f^{(l)}(y^{(l-1)}, \mathcal{W}^{(l)}_k) $ maps from the dimension of $l^\text{th}$ layer's output to the dimension of the $\mathcal{W}^{(l)}$. The matrix representation of such an operator could simply be enormous, when the layer is fully connected.

As for the delta-rule, we have to examine what is back-propagated to the previous layer. As it is seen, the gradient of the post-layer operator with respect to the input (the output gradient) is needed. Omitting the anyways not comprehensive notation on weights:
\[ \nabla^*_{x^{(l)}} G^{(l)} = \nabla^*_{x^{(l)}} f^{(l)}(y^{(l-1)}, \mathcal{W}^{(l)}_k) \cdot \big(\nabla^* G^{(l+1)}\big)(y^{(l)}). \]
Take the case of $f^{(l)}$ being an affine linear mapping $f^{(l)}(x^{(l)}, \mathcal{W}^{(l)}) = \mathcal{A}x^{(l)} + b^{(l)}$, where the affine parameters are in function of the weight $\mathcal{W}^{(l)}$, like in the case of convolutional or fully connected layers. In that case $\nabla^*_{x^{(l)}} f^{(l)}(y^{(l-1)}, \mathcal{W}^{(l)}_k) = \mathcal{A}^*$. Consequently, in all linear layers error is backpropagated via the transposed operator. In case of convolutions with transposed convolutions.

\section{Solving inverse problems. Previous results} \label{sec:solving_inverse_previous}

In this section we are to present what the previous attempts of applying CNNs to solve linear inverse problems are, or more specifically to denoising tasks, i.e\ inverting the effect of shift invariant noise generators. Such operators, apart from the inherent noise of measurements, are many times ill-posed in the sense of Hadamard: solutions may not exist, it is not unique and the pseudo-inverse may be continuous but very unstable with a huge condition number. Classical algorithms rarely have the chance to explore the large density manifolds of ideal reconstructions and, hence, neural networks become an alternative.

\subsection{Previous results}

The authors of \cite{MJU17} provide a general overview of this topic and argue that a good trade-off between following the learning approach and retaining classical elements is mandatory for success. One of their most important conclusions is that convolutional neural networks must be fit for denoising problems. On one hand, hand-crafting methods that grasp the essential representation in incomplete and noisy measurements to invert it back to an ideal reconstruction is a difficult approach, while CNNs are capable of learning manifolds. On the other hand, denoising should be a shift invariant operator and because of the approximation property of CNNs, they should be considered. 

Another point they establish is that in a linear inverse problem, given a measurement operator $H$, it is convenient to preprocess measurements before feeding it into the network using the backprojection $H^*$ in the case, when $H^*H$ becomes a shift-invariant operator. It is also conceivable that these inverse problems are the only ones that behave robust when solved via CNNs. This is especially relevant for us, because in subsection \ref{subsec:discrete_LI} it was discussed how the discrete $\Rad$ operator has this property. Besides that it was also shown that even $\Rad^+\Rad$ is an LI operator. This perfectly explains why most of the neural network based CT reconstruction methods choose to learn a regression from the FBP reconstruction rather than a sinogram. Obviously, in case of $\Rad$, using the measurements would also be impeded by the fact that geometrically related measurements are spatially found on a sinusoid shaped curve rather than in a vicinity of each other. 

Furthermore, they, too, reason that residual learning should be preferred over direct feed forward, i.e\ a hypothesis class of the form $\text{Id} + \text{CNN}_\mathcal{W}$ should be able to behave better, because this way the system is revolving around an almost identity operator and it is no longer obligatory to learn even such a simple task.

Almost the same team in \cite{JM17} roughly argue that the Landweber iterative method admitted by the $\Rad f = g$ inverse problem, which we presented in subsection \ref{subsec:Landweber}, directly translates to a fully convolutional neural network. The formulation in \eqref{it:Landweber_2}: $f^{(k+1)} = (\I - \omega\Rad^*\Rad) f^{(k)} + \omega\Rad^*g$ consists of a convolution operation $ \I - \omega\Rad^*\Rad $ and a bias $ \omega\Rad^*g $. Consequently, performing a number of iteration steps is equivalent to feeding the initial hypothesis into a feed forward CNN with specific weights.

A very interesting approach is outlined in \cite{GJ18}, which strongly influenced our research work, as well. The authors claim and justify that an iterative method should be applied in conjunction with a projection operator. It is suggested that every Landweber-iteration \eqref{it:Landweber} should be followed by an operator that projects the hypothesis back to the manifold of ideal reconstructions. Thus, they also provide an algorithm that certifies the correctness of a reconstruction, since the projection has its range as an invariant subset. The paper goes on explaining the learning scheme. The CNN is learnt via an ensemble of data points consisting of ideal reconstructions, FBP reconstructions and first order network outputs. The theoretical introduction also presents a more sophisticated algorithm and a corresponding theorem guaranteeing convergence of the procedure.

For CT reconstruction \cite{HW18} adopted the U-Net with slight modifications. Their network architecture is sketched by Fig.~\ref{fig:unet_Huang}. The network was trained to denoise already prepared reconstructions with low SNR and they argue that it is easy to create adversarial examples. These are input images that seem to follow the input distribution of the learning system, yet a confined amount of well-prepared noise or specially applied modifications cause the system to generate predictions with large deviance from the expected value. They prove the lack of robustness for cases when input images are contaminated with Poisson noise.

\begin{figure}[hp]
	\centering
	\includegraphics[width=1.0\textwidth]{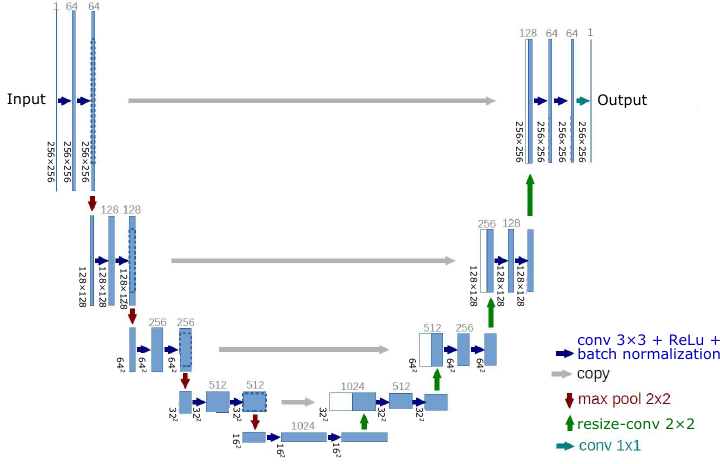}
	\caption{U-Net architecture applied by \cite{HW18} for CT reconstruction.}
	\label{fig:unet_Huang}
\end{figure}

This article encouraged \cite{HP19} to continue investigations in the direction of CNN-based CT image reconstruction denoising. Using the very same network as above, they devise an iterative algorithm for reconstruction. The neural network's output is post-processed to be consistent with the measured data, i.e.\ $\norm[2]{Rf - g} < \varepsilon$.
In order to move the output of the neural network even more towards the ideal reconstruction, an algebraic technique is invoked in this case (they used SART). The writers of the report move on to apply a standard total variation minimisation combined with SART.

Another study that has motivated us is \cite{HYJ16}, which reasoned that artifacts created during the use of classical algorithms, without being supplemented by either compressive sampling or neural network techniques, follow a distribution dependent on the scanning geometry. The most relevant factor here is the number of projections. As discussed in \cite{Nat01} and by us in section~\ref{sec:discretised_FBP}, FBP admits ideal reconstructions on the account of obeying sampling theoretical conditions. This condition is based on the number of angles projections are taken from, the number of parallel beams in one projection and the distance between slices. Hence, the ultimate goal to reduce the number of projections (and the emitted radiosity) depends on whether such constraints could be overcome. One such direction is recognising this geometry-dependent distribution. \cite{HYJ16}) proposes using other metrics (Betty-numbers) and attempt to give a theoretical reasoning that this distribution is actually very simple. The authors of the paper arrive to the conclusion, that the base U-Net should be taught to reproduce the difference between the ground truth and the FBP estimated reconstruction, because this could be an easier regression task.

Let us present other attempts as well. \cite{KMJ17} decomposes the image, prior to feeding into the network, with directional wavelets into many images. This could be considered motivated by compressed sensing, since this is a rewriting of the image in a representational basis. What follows is that these images falling in different linear subspaces are concatenated and fed in the network. The output of the network is also a decomposition in the same directional wavelet basis. The authors argue that this way the regression task becomes more tractable. A different reconstruction scheme proposed in \cite{WG16} is preparing dense layers that fully correspond to operators used in FBP. More precisely, for example, they implement the back-projection operator as a dense layer in their neural network. Undoubtedly, they receive, hence, a concise neural network. On the other hand, it could be argued that this approach could only reach its full potential if weights are regularised in a way that to each output pixel only truly relevant projections may contribute. Another promising direction is described by \cite{KH18}. They designed an architecture consisting of consecutive instances of a neural network and a data consistency layer. Nevertheless, it is not treated as an iterative algorithm, the number of modules is fixed prior to training and the full system is trained from the input of the first NN instance till the output of the last module. \cite{PH18} aimed to reconstruct the blurry, low SNR inputs in a super-resolution fashion. They argue that the number of projections could also be reduced by cutting down the parallel slices along the body axis. They feed their system with averaged, consecutive patches created from five images and attempt to reconstruct the middle one of these.

\subsection{Categorisation} \label{subsec:categories}

Until now we have presented some previous attempts of solving the inverse problem of CT reconstruction via the assistance of CNNs. To summarise, two major categories of such methods exist.
\begin{enumerate}
	\item[C1.] First, denoising CNNs. In these cases a CNN is used combined with the classical reconstruction algorithm, the FBP. Again, this was motivated by a number of facts: $\Rad^+\Rad$ is a shift invariant operator and convolutional neural networks are universal approximators on LI systems; deconvolution uses local information and sinograms contain geometrically related information in distant spots. Hence, the overall reconstruction scheme starting from the sinogram had the form $\text{CNN}_\mathcal{W} \circ \Rad^+$.
	\item[C2.] Second, the projected gradient descent method described by \cite{GJ18}. In this case, the reconstruction is based on an algebraic method, in particular the Landweber-iteration, but after each step the actual hypothesis is projected to the manifold of realistic reconstructions via a CNN.
\end{enumerate}

The following chapters, Chapter~\ref{chap:convnet} and \ref{chap:unrolled}, are going to present two from our lines of investigation in the field of CNN-based CT image reconstructions. In Chapter~\ref{chap:convnet}, an approach closer to category C1.\ is presented with the addition of consistency and sparsity promoting regularisation techniques. This method is called the \emph{measurement-consistent, sparsifying postprocess-ConvNet}. Meanwhile, Chapter~\ref{chap:unrolled} explores a previously uncovered direction of using CNNs in an iterative fashion for aiding reconstruction. The idea is close to category C2., but the neural regulariser is rather interpreted as part of the iterative refinement scheme. The denomination of this method was chosen to be \emph{unrolled support-kernel iterative regulariser GD}.

\chapter[Postprocess-ConvNet]{Measurement-consistent, sparsifying postprocess-ConvNet} \label{chap:convnet}

\TODOm{completely rework}

The idea of post-processing, denoising type CNNs has been discussed in detail in section~\ref{sec:solving_inverse_previous} and denoted as a major CNN-based reconstruction category (C1.) in subsection~\ref{subsec:categories}. The reconstruction scheme starting from the sinogram takes the form $\text{CNN}_\mathcal{W} \circ \Rad^+$. The objective function for such a learning problem reads 
\[\mathscr{O}(f, \bar{f}, \mathcal{W}) = \frac{1}{2} \big\lVert\text{CNN}_\mathcal{W} f - \bar{f}\,\big\rVert_2^2,\]
where $\bar{f}$ is the ideal, expected reconstruction corresponding to the FBP, low-quality reconstruction $f = \Rad^+ g$. Our addition consists of a measurement consistency promoter term and a sparsity promoting term based on classical total variation based regularisation techniques.

Section~\ref{sec:reg_ct} briefly goes through the classical theory of total variation minimisation. Afterwards, section~\ref{sec:convnet_design} presents the design considerations of the new method. Last, section~\ref{sec:convnet_experimentation} displays the experimental results about the performance of the system.


\section{Theory. Classical regularisation} \label{sec:reg_ct}

As already explained, regularisation is applied to coerce a priori information in cases when it is difficult for a learning system to retrieve perfect representations and features. In this section, a brief introduction to some classical prior information based techniques are shown. The reader is referred for a more detailed discussion to the textbook \cite{Tar05} and for an authentic introduction to \cite{Don06}.

We would like to solve the following optimisation problem:
\begin{IEEEeqnarray*}{rcl}
	\argmin_{f} & \quad & \big\lVert g-Rf \big\rVert_p^\gamma + \lambda \cdot \Omega(f),
\end{IEEEeqnarray*}
where the first term to be minimised, $\norm[p]{g-Rf}^\gamma$, called the fidelity term, forces the solution to be consistent with the input projection images. The second term, $\Omega(f)$ is a further constraint to be minimised. This could be for instance the negated distribution of the $f$ term, or any other constraint that we would like ensure that the solution will satisfy. The parameter $\lambda$ takes the role of creating the trade-off between the fidelity term and the regularisation term. The larger the $\lambda$, the more we satisfy additional constraints, but the more we neglect fidelity to the measured data. The optimization can also be treated as a Maximum a Posterior estimation problem (after applying logarithm transformation to the posterior probability function).

As for medical computing and, in general, image processing is concerned, a frequently used regularisation term is provided by the smoothness constraint. Real life images tend to minimise the number of abrupt changes. They usually feature relatively large homogeneous areas with very flat gradients and only edges cause significant gradients in colour space. For such analysis, one may check \cite{Pra07}. A usual regularisation function that enhances smoothness is the total variation norm.

\subsection{Total Variation minimisation} \label{subsec:TV}
The total variation of a picture $x=[x_{i,j}]_{i=\overline{1, M}, j=\overline{1,N}}$ is a matrix with size $(M-1) \times (N-1)$ with elements forming the amplitude of the discrete gradient of $x$:
\begin{equation}
	TV(x)_{i,j} \defeq \sqrt{ |x_{i,j}-x_{i+1,j}|^2 + |x_{i,j} - x_{i,j+1}|^2 }.
\end{equation}
More often we only use the L1-norm of the total variation operator:
\begin{align}
	\norm[1]{TV(x)} \defeq & \sum_{i=1}^{M-1} \sum_{j=1}^{N-1} \sqrt{ |x_{i,j}-x_{i+1,j}|^2 + |x_{i,j} - x_{i,j+1}|^2 }. \label{eq:TV1}
\end{align}
Hence, we can define the problem:
\begin{IEEEeqnarray}{rcl}
	\argmin_{f} & \quad & \big\lVert g-Rf \big\rVert_2^2 + \lambda \cdot \big\lVert TV(f) \big\rVert_1.
	\label{opt:TV}
\end{IEEEeqnarray}
In order to solve this, we present shortly the methodology proposed by \cite{BT09}. The authors adapt the \emph{proximal map} method introduced by \cite{Mor65}, which, given $h$, possibly non-convex function to be minimised and a point $x$, finds another point close to the initial one, but which reduces the value of $h$:
\begin{equation}
	\prox_t(g)(x) \defeq \argmin_u \Big\{ h(u) + \frac{1}{2t}\cdot\norm[]{u-x}^2 \Big\}.
\end{equation}
After this, a two step iterative method is suggested to optimise the problem~\eqref{opt:TV}. In the first step they optimise the fidelity term by the means of a gradient descent step. Secondly they minimise the regularisation term by applying the proximity operator. Hence, we arrive to the formulation:
\begin{align}
	f_k &= \prox_{t_k}(\norm[1]{TV}) (f_{k-1} - t_k \cdot \nabla^* \mathcal{J}(f_{k-1})) \nonumber \\
	&= \argmin_u \Big\{ \big\lVert TV(u) \big\rVert_1 + \frac{1}{2t_k} \cdot \big\lVert u -\big(f_{k-1} - t_k \cdot \nabla^* \mathcal{J}(f_{k-1})\big) \big\rVert^2 \Big\},
\end{align}
where $\mathcal{J}(x) = \norm{g - Rx}^2$ and, hence $\nabla \mathcal{J}(x) = 2Rx$. The method could also be supplemented by the use of the Nesterov-momentum, introduced in \cite{Nes83}, after which we immediately arrive to the Fast Iterative Shrinkage/Thresholding Algorithm (FISTA). For more details, check \cite{BT09} and \cite{BT09b}.

\subsection{Non-local Total-Variation} \label{subsec:NLTV}
Another, newer and less known method is the non-local total-variation method. TV-based methods tend to homogenise the entire image rather than force smoothness locally. In order to still reduce noise on the image, but keep the smoothness constraint as local as possible, we may define a non-local total-variation norm-function. We introduce this in the followings, based on the paper of \cite{KC16}:
\begin{align}
	\norm[1]{NLTV(u)} =& \sum_i \sqrt{ \sum_{j \in \Lambda_i } w_{ij}(u_j - u_i)^2} \nonumber \\
	\text{where } w_{ij} =& \exp \left( -\frac{\sum_{k=-a}^{a} G(k) \cdot |u(i+k) - u(j+k)|^2}{2 h_0^2} \right) \label{eq:NLTV-def}.
\end{align}
Let us clarify the meaning of the parameters: the non-local approach is reflected by the parameter $\Lambda_i$, which is a region of interest around voxel $i$. Parameter $w_{ij}$ introduces a weighting of the neighbouring points. As shown by the definition, we actually use a Gaussian kernel with a size of $(2a+1) \times (2a+1)$ and convolve it with difference of intensities in the neighbourhood. All-in-all, it is a complicated formulation, but the key aspect of it is its non-local approach.

Also, the authors promoted the use of a reweighted L1-norm in order to approximate the L0-norm. What they introduced is:
\begin{equation}
	\norm[RWL1]{x} \defeq \sum_i \frac{x_i}{x_i + \delta}. \label{eq:RWL1}
\end{equation}
Introducing an L0-approximating norm is related to the recent advances in compressive sampling theory.

The means of executing the algorithm could be discussed in a very detailed manner, or it could be checked out in the corresponding paper. Nevertheless, we only present the main idea and how it correlates with the previously introduced gradient descent based approach combined with the proximity operator used to solve the total variation problem in subsection~\ref{subsec:TV}. Three main steps are iterated different number of times:
\begin{align*}
	f_{k+1} \defeq& f_k + R^{+}(g_k - Ru_k) & \text{(any ART update)} \\
	u_{k+1} \defeq & \argmin_u \left\{ \gamma \norm[RWL1]{NLTV(u)} + \frac{1}{2}\norm{u-f_{k+1}} \right\}\\
	g_{k+1} \defeq & g_k + (g - Ru_{k+1}).
\end{align*}
The first step optimises the fidelity term by using not necessarily one iteration of an arbitrary algebraic reconstruction technique. The method is even more sophisticated, since the output data and subtracted images are from a sequence that are also being iterated through time. The second formula reflects on the minimisation of the regularisation term and we may easily recognise the proximity operator here.

\subsection{Compressive sampling}
Compressive sampling or compressed sensing (CS) deals with the analysis of how many measurements we must take in order to reconstruct the original image within a predefined error term. As already motivated, reducing the number of measurements reduces the amount of radiation received by the body of the patient. CS theory actually analyses whether there exists a representational basis $\Psi$ such that the representation of the image $f$ in $\Psi$ is sparse enough, i.e.\ the number of non-zero components is low. Until now our task was to solve the inverse problem $g = Rf$ with some additive noise. Let us assume that we do posses a ``good'' basis $\Psi$. In that case, the solution to the following optimisation problem may be close to the exact solution or it may be exact with very high probability:
\begin{equation}
\begin{IEEEeqnarraybox*}[][c]{r'l}
	\min_{\tilde{f}} & \norm[1]{\tilde{f}} \\
	& g = R\Psi\tilde{f}.
\end{IEEEeqnarraybox*}
\label{opt:CS}
\end{equation}
This turns out to be a lucrative theorem, since problem~\eqref{opt:CS} is a convex optimisation problem, the L0-norm is relaxed to L1-norm. 

\TODOm{Candes Wakin article}

\section{Architectural design} \label{sec:convnet_design}


%
As it was already stated in section \ref{sec:solving_inverse_previous}, \cite{HYJ16} argued on a theoretical level that for our problem it pays off to have the neural network learn the artifacts created during the reconstruction algorithm due to insufficient number of projections. This noise contaminating the image follows, according to our conjecture, a distribution entirely dependent on the geometry of the scanning arrangement. The authors proved that the representation of this type of noise is very simple and, hence, tractable for regression models.
The CNN $F(\cdot, \mathcal{W}) = F_\mathcal{W}$ is, then, formulated as a residual mapping $F(\cdot, \mathcal{W}) = \I + \mathcal{C}(\cdot, \mathcal{W})$, where $\mathcal{C}(\cdot, \mathcal{W}) = \mathcal{C}_\mathcal{W}$ is an autoencoder-style convolutional network, most typically U-Net. (As already explained, it is convenient to provide a straightforward implementation of the identity mapping.)


As many reports, our work has, too, involved the use of the U-Net architecture for $\mathcal{C}(\cdot, \mathcal{W})$, previously depicted by Fig.~\ref{fig:unet_Huang}, but we changed the system of batch normalisation layers. Ronneberger et al.~\cite{RFB15} did not take advantage of batchnorm layers, Huang et al.~\cite{HW18} applied much more frequently, after each and every convolutional layer. We decided that the most critical points of network where local data normalisation would be helpful are directly after the max pooling layers and after concatenation layers. See Fig.~\ref{fig:unet_target}.

\begin{figure}[htp]
	\centering
	\includegraphics[width=0.8\textwidth]{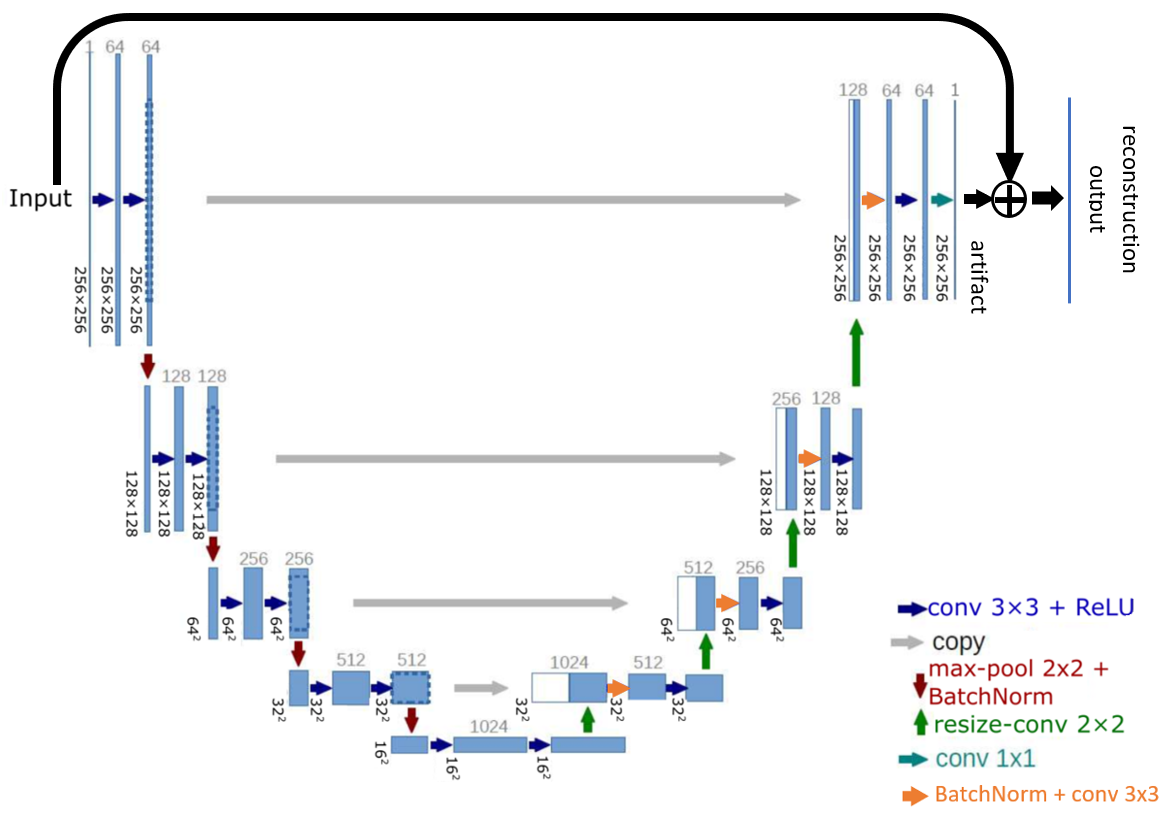}
	\caption{Target network $F_\mathcal{W}$ applied in our project. The major difference compared to the adopted U-Net from~\cite{HP19} is that a further addition of the output and the input are implemented, i.e.\ there is a residual connection. This ensures that the original U-Net input learns the artifacts, which could be an easier regression task.}
	\label{fig:unet_target}
\end{figure}


\subsection{Reconstruction fidelity}

Let $\loss ^2$ denote a mean squared error between two elements, i.e.:
\begin{equation*}
	\loss^2 (\tilde{f}, \bar{f}) = \dfloor[\big]{\tilde{f} - \bar{f} \,}_2^2,
\end{equation*}
where $\dfloor{x}_p^p \defeq (1/\dim{(x)}) \norm[p]{x}^p$ is the notation for the mean $p$-power error, i.e.\ the momentum of order $p$. We applied two losses on this reconstruction output layer, namely a mean squared error between the predicted and expected output reconstructions (denoted by $\bar{f}$) and a total variation distance between the two images. The latter loss function was motivated by the experience that edges tend to become blurry and contrast-to-noise ratio is diminished. The loss function on this layer involving a weight parameter $\gamma$ is:
\begin{equation}
	\loss^2 (\tilde{f}, \bar{f}) + \gamma \dfloor[\big]{TV(\tilde{f}-\bar{f}\,)}_1.
\end{equation}

\subsection{Measurement consistency}

Besides that, the network is fitted with a data consistency layer, the Radon transformation layer, which will transform the reconstruction output into its sinogram and will fit it to the expected sinogram. Naturally, the reconstruction output layer could be enough for training in the sense that any optimal weight value for the case when we assign non-zero loss only to the reconstruction output layer would be optimal for the Radon-layer, as well. The difference is that adding an extra objective function with the same optimum changes the overall objective function, hence the surface on which optimisation happens. It could be argued that a fortunate selection of loss functions and loss weighting could ``convexify'' the cost function, i.e.\ help the optimisation, however that remains for future research. But it is doubtless that in this case the neural network's weight configuration is only allowed to traverse on a path that continuously preserves consistency with measured projections. In this sense this objective function aids the weights to find the better path to the optimum. A mean squared error was also applied to the output of this layer:
\begin{equation*}
\loss^2 (\Rad \tilde{f}, \Rad \bar{f}) = \dfloor[\big]{\Rad \tilde{f} - \Rad \bar{f}}_2^2.
\end{equation*}

\subsection{Sparsity}

We further need to enhance the sparsity of total variation of the image. Once again, techniques in compressed sensing allow us to compute the reconstruction from far fewer measurements. According to \cite{CW08b} sparsity could be achieved through iterative reweighted total variation minimisation. They suggest the use of the following optimisation problem:
\begin{equation}
\begin{IEEEeqnarraybox*}[][c]{r'l}
\min_{x \in \R^n} & \sum_{i=1}^n \ln(|x_i| + \varepsilon)\\
& y=\Phi x,
\end{IEEEeqnarraybox*}
\label{eq:log_min}
\end{equation}
where $\Phi$ is an arbitrary operator applied on the linear space $\R^n$. This problem assumes that $\Phi$ is a perfectly sparse representational basis for the data $y$. They further go on to establish an upper bounding minimisation problem for the sake of iterative algebraic reconstruction techniques. Since our neural system differs from usual iterative total variation algorithms in the sense that we are optimising in a continuous space using gradient descent-type approaches, we decide to use this logarithmic formula directly as a loss function:
\begin{equation*}
	\logloss (\tilde{f}\,) =\dfloor[\big]{\ln\big(TV(\tilde{f}\,)_{i,j} + \varepsilon\big)\big|_{i,j}\,}_1 = {\footnotesize\frac{1}{\dim(\tilde{f}\,)}} \sum_{(i,j)} \ln(TV(\tilde{f}\,)_{i,j} + \varepsilon).
\end{equation*}

\subsection{Overall loss function and optimisation}

Hence, the complete loss function of our optimisation scheme is the following:
\begin{equation*}
	\Loss (\tilde{f}, \bar{f}\,) = \tau_1 \left( \loss^2 (\tilde{f}, \bar{f}\,) + \gamma \dfloor[\big]{TV(\tilde{f}-\bar{f}\,)}_1 \right) + \tau_2 \loss^2 (\Rad \tilde{f}, \Rad \bar{f}\,) + \tau_3 \logloss (\tilde{f}\,).
\end{equation*}

In order to keep kernel weight values in low domains, an L2-Tikhonov kernel regularisation was applied on the convolutional layers. Therefore, the generalised error optimisation problem that should be addressed is:
\begin{IEEEeqnarray}{r'l}
	\min_\wei & \meanlim_{(f, \bar{f}) \sim \mathcal{D}_{f\times\bar{f}}} \left[\Loss \big(F_\wei (f), \bar{f}\,\big)  + \tau_4 \dfloor{\wei_\mathcal{K}}_2^2\right],
\end{IEEEeqnarray}
where $\mathcal{D}_{f\times\bar{f}}$ is the joint distribution of images reconstructed from realistic, noisy, undersampled sinograms, alongside with the corresponding ideal, high-quality reconstructions.

\section{Experimentation} \label{sec:convnet_experimentation}

\subsection{Dataset. Preprocessing the data} \label{subsec:data_preprocess}

As a dataset, the image database published by the Lung Image Database Consortium (LIDC) and Image Database Resource Initiative (IDRI) was used (see the publications~\cite{LIDC11}, \cite{LIDC11data} and the webpage~\cite{LIDC11dataweb}, where the database may be accessed).

The raw data consists of approximately $250$ thousand $512 \times 512$ reconstructions belonging to $1010$ patients. Outside the training loop, offline, these images are downsampled to $256 \times 256$ due to memory limitations and for the sake of complexity reduction. Afterwards, their Radon-transforms with the priorly prescribed number of projections are produced. Results here are based on sinograms with $40$ projections spread sparsely, uniformly in the angle range $[0\degree, 180\degree]$.
As the main goal of CNN-based CT reconstruction is the substantial reduction in the number of projections, it is noteworthy that commercially available infrastructures (using FBP as reconstruction algorithm) acquire a number of projection in the range of few hundreds, possibly reaching close to a thousand.

At training time, the $256 \times 256$ downsampled reconstruction is used as ideal, expected output reconstruction, the persisted sinograms represent the ideal, expected output sinogram. The input of the neural network is generated real time by adding noise to the ideal sinogram and reconstructing it using the FBP algorithm. The noise model is the one adopted from~\cite{KC16}. It assumes a normal distribution noise corresponding to thermal noises and afterwards applies the usual Poisson-noise:
\begin{equation}
\begin{aligned}
	\mathcal{I} & = \text{Poi} \left( \mathcal{I}_0 e^{-\Rad f} + \mathcal{N}(0, \sigma \mathcal{I}_0) \right)\\
	\widehat{\Rad f} & = - \ln \frac{\mathcal{I}}{\mathcal{I}_0},
\end{aligned}
\label{eq:noise_model}
\end{equation}
where $\mathcal{I}_0$ represents the input amount of radiation expressed in intensity. Matrix and vector operations should be understood element-wise. After linearisation, the value $\widehat{\Rad f}$ is reconstructed using FBP and the image is prepared for being an input. The parameters $\mathcal{I}_0$ and $\sigma$ are set such that the SNR of noise on sinograms is around $40$ dB.

\TODOm{effect of components on SNR}

\subsection{Hyper-parameter setting}
The loss contains six different parameters: $\varepsilon$, $\gamma$, $\tau_1$, $\tau_2$, $\tau_3$, $\tau_4$. It remains an open question, what the correct relationship of these six parameters could be. During our training, some of them were changed from time to time, depending on the actual state of losses and metrics, hence it is plausible that proper scheduling is needed. Nonetheless, we conjecture that the system is robust against changing these parameters, hence a roughly appropriate setting should exist and should be able to achieve optimal weights. Nevertheless, it is noteworthy that the Radon-transform of an image contains values that are sums along lines in the input image, hence the ratio of pixel values are theoretically and practically also around the size of the image length. Hence we chose to have $\tau_2$ somewhere around the inverse of the image side length. The parameter $\tau_1$ would be set around $10$ to $1000$ depending on the actual state of learning. During training we used Adam optimizer with learning rate decreasing from $10^{-3}$ to $10^{-5}$. However, we reached almost minimal reconstruction output loss during the $10^{-3}$ learning rate phase. Last, but not least, a L2-kernel regularisation applied to convolutional layers had weight parameter $\tau_4$ set to values between $10^{-3}$ and $10^{-4}$.

\subsection{Results} \label{subsec:results_convnet}

The test dataset applied for all following results consists of entire patient datasets, which have never been shown to the network during training and validation iterations. At some figures we provided two important and standard metrics, the structural similarity index measure (SSIM) and the mean absolute error, but calculated in Hounsfield Unit (MAE [HU]).

The structural similarity index measure (see \cite{WB04}) calculates the following index on multiple windows of two images and averages the results:
\begin{equation*}
	SSIM(X, Y) = \frac{(2\mean(X)\mean(Y) + c_1)(2 cov (X, Y) + c_2)}{(\mean^2(X) + \mean^2(Y) + c_1)(\sigma^2(X) + \sigma^2(Y) + c_2)},
\end{equation*}
where $c_1$ and $c_2$ are constant parameters set based on the dynamic range of images. The authors' original suggestions in setting these parameters were followed. It is easily proven that on any window this index is between $-1$ and $1$ and its value is $1$ if and only if the means are equal and the correlation between the two images is $1$, hence they are equal with probability 1.

The mean absolute error is converted to Hounsfield for three reasons. First, the linear attenuation coefficients depend on the intensity of radiation that the body is exposed to, hence the value does not hold too much information if the scanning device's calibration is not known. Secondly, the dataset itself contained Hounsfield Units, which we transformed to attenuation coefficient with our own scaler. Therefore it is advised to transform it back. Moreover, Hounsfield Unit was documented exactly for the reason of unifying CT scanning measurements with respect to scanning devices and radiation intensity choice. For the sake of comparison, the radiosity expressed in HU of bones varies anywhere between $300$ and $1800$, the same for lungs varies approximately between $-1000$ and $-600$. The airs radiosity is always calibrated to $-1000$HU. The smallest change of radiosity that could hold clinical information in the case of CT scanning should be considered around a few 10s of HUs.

\subsubsection{Evaluation}

Table~\ref{tab:results_convnet} summarises the evaluation results for the measurement-consistent, sparsifying postprocess-ConvNet (\emph{MC-S-P-ConvNet}) that we have presented in this chapter. For comparison, the same metrics for FBP under the same circumstances are also provided. Displayed are the mean absolute error expressed in Hounsfield Units (\emph{MAE [HU]}), the structural similarity index measure (\emph{SSIM}), the signal to noise ratio expressed in dB (\emph{SNR [dB]}), the relative root mean squared error between the received and ideal reconstructions (\emph{RelError}) and the relative root mean squared error in the measurement space (\emph{RadonRelError}). It is easily concluded that our post-processing method is far superior to the vanilla FBP algorithm.
\begin{table}
	\centering
	\begin{tabular}{ccc}
		\hline
		& FBP & MC-S-P-ConvNet  \\
		\hline
		MAE [HU] & 272.9 & 38.80  \\
		SSIM & 0.362 & 0.924  \\
		SNR [dB] & 6.00 & 20.50  \\
		RelError & 0.530 & 0.10 \\
		RadonRelError & 0.083 & 0.012\\
		\hline
	\end{tabular}
	\caption{Evaluation results for the measurement-consistent, sparsifying postprocess-ConvNet (MC-S-P-ConvNet) in comparison with the same metrics for FBP under the same circumstances.}
	\label{tab:results_convnet}
\end{table}


\subsubsection{Case-studies}

After training the network we randomly chose a few images from the preserved test dataset. Again, this test dataset consists of entire patient datasets, which have never been shown to the network during training and validation iterations.
In figures~\ref{fig:00_recs} to~\ref{fig:07_recs} we depict different slices of axial CT reconstructions from different patients. In each triplet the middle image represents the ground truth, the desired reconstruction. On the left of this is found the generated reconstruction that is fed to the neural network and which was computed by using the Filtered Back Projection on the noisy sinogram containing only 40 projections. Finally, to the right the reconstruction predicted by the neural network may be viewed. For the sake of simplicity, MAE was calculated for the entire image. The reason for being negligent with this metric (or metrics in general) is that a really powerful metric would be a measure that is sensitive to regions where large changes in radiosity occur and which are important from the diagnostic perspective. This is still actively studied by us. As far as the reconstructions are concerned, most of them are of acceptable quality with MAE in HU being under $35$.

\newcommand{\myfig}[3][]{
	\begin{minipage}{0.23\textwidth}
		\centering
		\includegraphics[width=\textwidth, keepaspectratio]
		{images/FBN_MSE_TV_loss_Training.50-0.0301-20210328-230950/#2/#3.png}%
	\end{minipage}%
}
\newcommand{\figrow}[5]{
	\begin{figure}[hp]
		\centering
		\myfig[]{#1}{in_noisy_rec}
		\myfig[]{#1}{out_expected_rec}
		\myfig[]{#1}{predicted_rec}%
		\begin{minipage}{0.24\textwidth}
			\centering
			\begin{minipage}{0.8\textwidth}
				\centering				
				\captionof{figure}{}
				\label{fig:#1_recs}		
				\small
				\begin{tabular}{cc}
					\hline
					MAE [HU] & #2 \\
					SSIM & #3 \\
					SNR & #4 \\
					RelErr & #5 \\
					\hline
				\end{tabular}
			\end{minipage}
		\end{minipage}%
	\end{figure}
}

\newcommand{\cor}{\vspace{-0.75cm}}

\figrow{00}{28.30}{0.954}{22.34}{0.076}\cor
\figrow{01}{44.81}{0.913}{18.55}{0.118}\cor
\figrow{02}{27.09}{0.958}{22.46}{0.075}\cor
\figrow{03}{44.99}{0.909}{18.21}{0.122}\cor
\figrow{04}{46.91}{0.915}{18.32}{0.121}\cor
\figrow{05}{37.38}{0.924}{20.91}{0.090}\cor
\figrow{06}{22.49}{0.968}{24.98}{0.056}\cor
\figrow{07}{62.37}{0.880}{18.66}{0.164}\cor

%
%

\chapter[Unrolled iterative GD]{Unrolled support-kernel iterative regulariser GD} \label{chap:unrolled}


The idea of the iterative, unrolled application of a CNN for CT image reconstruction was sketched in subsection~\ref{subsec:categories} and was assigned the category C2.
Our new approach is similar to this idea. It is desirable to perform necessary corrections by a manifold learning system. Nevertheless, instead of performing a projection onto the manifold, the neural network should become part of the iterative refining system. Note that the Landweber-iteration converges to $\Rad^+ g + \proj_{\ker\Rad}f^{(0)}$. We will prove this convergence once again via contractive functions. However, this method is unable to change the kernel space component, once it has been set by the initial estimate $f^{(0)}$. The kernel space component should be adjusted by the neural network, also in a contractive manner. In order to make the neural network contractive in the kernel space, a directional input gradient regularisation will be imposed on the optimisation objective. Once provided everything, the final network will be applied iteratively on our reconstructions in conjunction with simple ART-steps.

Section~\ref{sec:unrolled_theory} presents the theoretical background for contractive functions. Following that, section~\ref{sec:unrolled_design} outlines our algorithm and objective function. Further regularisation aspects are touched, too. Finally, section~\ref{sec:unrolled_experimentation} contains the experimentation with the method together with the results.

\section{Theory. Contractive functions. Fixed-point theorem} \label{sec:unrolled_theory}


Given a normed space $\mathcal{N}$, a function $\varphi : \mathcal{N} \rightarrow \mathcal{N}$ is called \emph{Lipschitz-continuous}, if there exists a $c > 0$ such that:
\[ \norm{\varphi(x)-\varphi(y)} \leq c \norm{x-y}. \]
If $c < 1$ holds, then besides being Lipschitz-continuous, $\varphi$ is also called \emph{contractive}.

\begin{theorem}[Banach's fixed-point theorem]
Let $\mathcal{N}$ be a Banach-space and let $\varphi$ be a contractive function on $\mathcal{N}$. In this case $\varphi$ has got a unique fixed point, i.e.\ a point $x^* \in \mathcal{N}$, where $\varphi(x^*) = x^*$. Besides, for any $x_0 \in \mathcal{N}$, the iteration $x_{k+1} = \varphi(x_k)$ is convergent and $x_k \rightarrow x^*$, as $k \rightarrow \infty$. Furthermore:
\[ \norm{x_k - x^*} \leq \frac{\norm{x_1-x_0}}{1-c} c^k. \]
\end{theorem}

Let us review the Landweber-iteration in the form given in equation \eqref{it:Landweber_2}: $f^{(k+1)} = (\I - \omega\Rad^*\Rad) f^{(k)} + \omega\Rad^*g$. For the function $L(f) = (\I - \omega\Rad^*\Rad) f + \omega\Rad^*g$ we obtain:
\[ L(f_1) - L(f_2) = (\I - \omega\Rad^*\Rad) (f_1-f_2). \]
Operator $Q = \I - \omega \Rad^*\Rad$ is invariant on the subspace $\ker\Rad$ and for $f = f_1-f_2 \in \supp \Rad$ we previously got:
\begin{align*}
\norm[2]{Qf}^2 &= \Big\lVert\sum_{i=1}^r (1-\omega\sigma_i^2)\langle v_i, f \rangle v_i\Big\rVert_2^2 = \sum_{i=1}^r \big|1-\omega\sigma_i^2\big|^2\, \lVert\langle v_i, f \rangle v_i \rVert^2 \leq \\
& \leq \max_{1 \leq i \leq r}\big|1-\omega\sigma_{i}^2\big|^2 \, \Big\lVert\sum_{i=1}^r\langle v_i, f \rangle v_i\Big\rVert_2^2 = \max_{1 \leq i \leq r}\big|1-\omega\sigma_{i}^2\big|^2 \, \norm[2]{f}^2.
\end{align*}
Hence, $\norm[2]{Qf} \leq \max_{1 \leq i \leq r}\big|1-\omega\sigma_{i}^2\big|^2 \, \norm[2]{f}$. 
The convergence criterium for the iteration was initially stated in \eqref{eq:condition_Landweber0} exactly as $|1-\omega \sigma_{i}^2| < 1$, for any $i$. Thus, in these cases the function defining the Landweber-iteration is contractive on $\supp\Rad$.

To see the overall convergence and limit of the Landweber-iteration (already proven in subsection~\ref{subsec:Landweber}, but now argued differently), note that for any $f^{(0)}$,
\[ f^{(k)} = L^k(\proj_{\supp\Rad}f^{(0)}) + \proj_{\ker\Rad}f^{(0)}. \]
As $L:\supp\Rad \rightarrow \supp\Rad$ is contractive, the right hand side is convergent and the Landweber-iteration is convergent in general. For the limit point, we see that
$f^* = \proj_{\supp\Rad} f^* + \proj_{\ker\Rad} f^{(0)}$
and
\begin{align*}
\proj_{\supp\Rad} f^* = L(\proj_{\supp\Rad} f^*) = \proj_{\supp\Rad} f^* - \omega\Rad^*\Rad \proj_{\supp\Rad} f^* + \omega\Rad^*g, \\
\Rad^*\Rad \proj_{\supp\Rad} f^* = \Rad^*g.
\end{align*}
Therefore, $\proj_{\supp\Rad} f^*$ is a solution of the normal equation belonging to $\Rad f = g$, which is only possible if $\proj_{\supp\Rad} f^* = \Rad^+ g$, and, thus,
$f^* = \Rad^+g + \proj_{\ker\Rad}f^{(0)}$.

\section{Overview of design} \label{sec:unrolled_design}

With all the derivations of section \ref{sec:unrolled_theory}, we conclude that the Landweber-iteration is applicable for support space reconstruction. Kernel space reconstruction on the other hand is done via a neural network. Denoting the convolutional network with weight vector $\mathcal{W}$ as $F(\cdot, \mathcal{W})=F_{\mathcal{W}}$, our system has the form:
\begin{IEEEeqnarray}{c"t"c}
f^{(k+1)} = F\left(L^s (f^{(k)}), \mathcal{W}\right) & or & f^{(k)} = (F_\mathcal{W} \circ L^s)^k (f^{(0)}).
\label{it:our_approach}
\end{IEEEeqnarray}
We perform $s$ ART-steps before each network call. A larger choice of $s$ is motivated by the too strong expressiveness of the neural network compared to a single Landweber-iteration. The CNN $F(\cdot, \mathcal{W})$ is formulated, as already explained in section \ref{sec:solving_inverse_previous}, as a residual mapping $F(\cdot, \mathcal{W}) = \I + \mathcal{C}(\cdot, \mathcal{W})$, where $\mathcal{C}(\cdot, \mathcal{W}) = \mathcal{C}_\mathcal{W}$ is an autoencoder-style convolutional network, most typically U-Net. (As already explained, it is convenient to provide a straightforward implementation of the identity mapping. This is especially true for our case, since typically it is not desirable that every iteration produces such a large step that the effect of the classical ART steps are eroded.) The initialisation $f^{(0)}$ could be arbitrarily chosen, but for the sake of reducing the complexity of the manifold our neural network is learning, it is advisable to have $f^{(0)}$ as a low-quality reconstruction. We chose $f^{(0)} = L^p(\underline{0})$. After having defined the architecture, it remains to be discussed how to force the neural network $F(\cdot, \mathcal{W})$ to be invariant on the support of $\Rad$ and how to establish the contractive property on the kernel space of $\Rad$.

The end-result will become a reconstruction scheme that we named \emph{unrolled support-kernel iterative regulariser GD}, because it alternates between augmenting support and kernel space components. Also GD, because the Landweber-iteration-step is in fact a GD step. For the architecture see Fig.~\ref{fig:iterative_model}.

\begin{figure}[ht]
	\centering
	\includegraphics[keepaspectratio,width=0.8\textwidth]{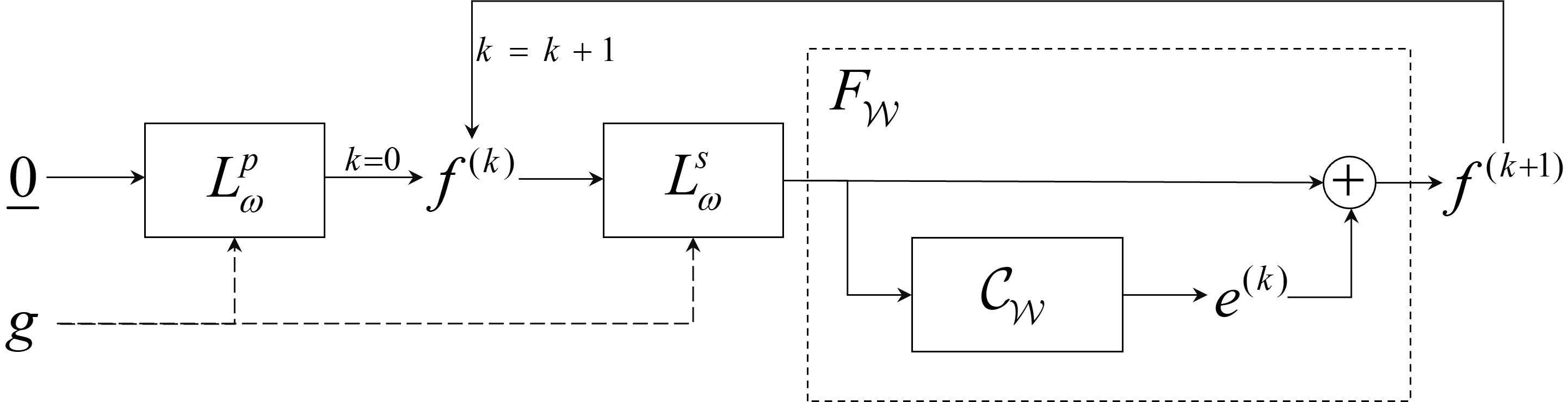}
	\caption{Block diagram of the \emph{unrolled support-kernel iterative regulariser GD}.}
	\label{fig:iterative_model}
\end{figure}

\subsection{Invariance on the support} \label{subsec:support_invaraince}

The learning has to incentivise the neural network to minimise changes in $\supp\Rad$ and concentrate on $\ker\Rad$. Strict invariance needn't be mandatorily asked, since measurement data may very well be noise-contaminated in the support space, and, therefore, a slight regularisation effect even in that subspace's components is well-seen. It is easy to incorporate such a incentivising term in the overall optimisation objective.  For an arbitrary $f$ reconstruction after the ART-steps, it is sought that changes made by $F$ mostly fall into the kernel space: $\Rad\big( F(f, \mathcal{W}) - f \big) \approx 0$. This is equivalent to $\Rad\big( \mathcal{C}(f, \mathcal{W}) \big) \approx 0$. Hence a regularisation term promoting changes in the kernel space is $\norm[2]{\Rad\big( \mathcal{C}(f, \mathcal{W}) \big) }^2$, where  again, $f$ is the output the ART-steps.

\subsection{Contractivity on the kernel. Input gradient regularisation} \label{subsec:input_gradient}

One further problem to tackle is to coerce Lipschitz-continuity and more specifically, contractivity of the neural network in the kernel space of $\Rad$. This may be done through directional input gradient regularisation. 

First, let us denote $H(f, \mathcal{W}) \defeq (F_\mathcal{W} \circ L^s)(f)$ the entire mapping for one complete iteration. It is well-known that if $H$ is partially differentiable in every point $f$, and if for any $f, \mathcal{W},$ we have $\norm[2]{\nabla_f H(f, \mathcal{W})} \leq c < 1$ for some constant $c$, then the overall mapping would become contractive. To see this, note that for any $f_1, f_2$
\begin{equation}
H(f_1, \mathcal{W}) - H(f_2, \mathcal{W}) = H(f_2 + t(f_1-f_2), \mathcal{W}) \big|_{t=0}^1 = \int_0^1 h'(t)dt,
\label{eq:IGR_1}
\end{equation}
where $h = \{t \rightarrow H(f_2 + t(f_1 -f_2), \mathcal{W})\}$. Denoting $f_t = f_2 + t(f_1 -f_2)$, we obtain $h'(t) = \nabla_f H(f_t, \mathcal{W}) \, (f_1-f_2)$. In this case, taking norms in \eqref{eq:IGR_1}:
\begin{align}
\norm{H(f_1, \mathcal{W}) - H(f_2, \mathcal{W})} = \norm{ \int_0^1 \nabla_f H(f_t, \mathcal{W})\,(f_1 - f_2) dt} =  \label{eq:IGR_2}\\
= \norm{ \int_0^1 \nabla_f H(f_t, \mathcal{W})dt \,(f_1 - f_2) } \leq \norm{\int_0^1 \nabla_f H(f_t, \mathcal{W})dt} \cdot \norm{f_1-f_2} \leq \nonumber\\
\leq \int_0^1 \norm{\nabla_f H(f_t, \mathcal{W})}dt \cdot \norm{f_1-f_2} \leq \max_{f, \mathcal{W}} \norm{ \nabla_f H(f, \mathcal{W})} \cdot \norm{f_1-f_2}. \label{eq:IGR_3}
\end{align}
Therefore, if $\max_{f, \mathcal{W}} \norm{ \nabla_f H(f, \mathcal{W})} < 1$, then $H(\cdot, \mathcal{W})$ is contractive. Consequently, it would be convenient to add $\norm[2]{\nabla_f H(f, \mathcal{W})}^2$ as a regularisation term to the learning objective. We emphasize, that L2-norm here is the operator(induced)-norm and not the euclidean norm. At the time being this is impeded by the fact that $\nabla_f H(f, \mathcal{W})$ is an operator that maps between linear spaces with dimensions equal to that of reconstructions. This operator is gigantic and we are not in with the chance of storing the tensor representation. All we need is the L2-norm, the largest singular value, which is also not easy to compute.

Therefore, instead of input gradient regularisation we choose directional gradient regularisation, i.e.\ we will operate with far smaller directional gradients and casual euclidean norm. Observe that for any function $h$ and direction $e$, we have $\nabla_e h(x) = \nabla h(x) \, e$, and therefore,
\[ \max_{\norm{e} = 1} \norm{\nabla_e h(x)} = \max_{\norm{e} = 1} \norm{\nabla h(x) \, e} = \norm{ \nabla h(x)}, \]
based on the definition the norm of a linear operator. Thus, $ \max_{x, \norm{e} = 1} \norm{\nabla_e h(x)} = \max_x \norm{ \nabla h(x)}$. To derive the directional input gradient penalisation term another way, we start off from \eqref{eq:IGR_2}. Let us denote the normalisation of vector $v$ by $\hat{v}$.
\begin{align*}
\norm{H(f_1, \mathcal{W}) - H(f_2, \mathcal{W})} &= \norm{ \int_0^1 \nabla_f H(f_t, \mathcal{W})\,(f_1 \!-\! f_2) dt} =\\
&=\norm{ \int_0^1 \nabla_f H(f_t, \mathcal{W})\,\widehat{f_1 \!-\! f_2}\, dt} \cdot \norm{f_1- f_2} \leq \\
&\leq \int_0^1 \norm{\nabla_f H(f_t, \mathcal{W})\,\widehat{f_1 \!-\! f_2}}\, dt \cdot \norm{f_1- f_2} \leq \\
& \leq \max_{f, \mathcal{W}, \norm{e}=1} \norm{\nabla_f H(f, \mathcal{W})\,e} \norm{f_1-f_2}.
\end{align*}
Our plan is to introduce a regularisation term $\norm{\nabla_f H(f, \mathcal{W})\,e}^2$. The computational advantage of this is clear, the application of the derivative on any vector is achievable on-the-fly by \emph{forward mode automatic differentiation}. As for a drawback, one obviously notes that penalisation only applies to a few directions and not all of them.

Now, for introducing this term one has to provide a specific unit vector $e$.
During our research work we decided to choose the normalised vector from the actual reconstruction to the ideal, expected reconstruction. In details: if $f^{(k)}$ is the current estimate, and $\bar{f}$ is the corresponding ideal reconstructions, then $e$ is considered $\widehat{\bar{f} \!-\!\! f^{(k)}}$. For the exact loss function and optimisation algorithm see the followings. Deciding on this directional vector is motivated by the fact that the neural network would anyways bring the current hypothesis closer to the ideal reconstruction and, thus, we are making changes along the direction $\widehat{\bar{f} \!-\!\! f^{(k)}}$.

\begin{remark}
There exist previous works analysing Lipschitz-continuity and contractivity. The authors of \cite{AGC20} propose a variational problem for learning activation functions that increase the capacity of neural networks, while also maintaining an upper-bound of the global Lipschitz-constant. The report \cite{SV18} overviews the possibilities of computing the global Lipschitz-constant of a network, while they state without proof that such a task becomes \NPh{} as soon as we have two layers. The report produces a modified power method for obtaining, again, an upper bound on the Lipschitz-constant. This work is more useful from a theoretic standpoint, it lacks a direct regularisation suggestion. Probably the most promising work in the field is \cite{GFP21}, which for the sake of enforcing Lipschitz-continuity, introduces a constrained learning optimisation problem and solves it via projected gradient descent method.
\end{remark}

\subsection{Optimisation. Loss function and objective} \label{subsec:it_loss}
For the optimisation objective, most of the regularisation terms have been defined in previous subsections. As for a direct fidelity term between the ideal reconstruction and a current estimate, a mean squared error is always applied. Our goal is to train the network in an unrolled manner, but augmenting the learning process not only at the final depth, but also at all inner stages. A final depth is preset as $D$ and the importance of iterations is pronounced by amplifying multipliers. Denoting with $\gamma_a$ a reconstruction fidelity amplifier, with $\bar{f}$ the ideal reconstruction, the reconstruction fidelity error for iteration depth $0 \leq k \leq D-1$ is:
\[ \gamma_a^k\cdot \dfloor[\big]{ H(f^{(k)}, \mathcal{W}) - \bar{f}\, }_2^2. \]
Here, $\dfloor{x}_p^p \defeq (1/\dim{(x)}) \norm[p]{x}^p$ is, again, the notation for the mean $p$-power error or the momentum of order $p$. The error term caused by undesirably large changes caused by the network in the support space, presented in subsection~\ref{subsec:support_invaraince} is the following. If again the system is modelled, as $H(\cdot, \mathcal{W}) = (\I + \mathcal{C}(\cdot, \mathcal{W})) \circ L^s$, then for iteration depth $k$, the support space error weighted by $\gamma_s$:
\[ \gamma_s \dfloor[\big]{ \Rad\big(\: H(f^{(k)}, \mathcal{W}) \!-\! L^s(f^{(k)}) \:\big) }_2^2 = \gamma_s \dfloor[\big]{ \Rad\big( \mathcal{C}_\mathcal{W}(L^s(f^{(k)})) \big) }_2^2. \]

As for the input gradient regularisation, in case of iteration depth $k$, the penalisation term derived in subsection~\ref{subsec:input_gradient} with control weight $\gamma_g$ is the following:
\[ \gamma_g \dfloor[\big]{ \nabla_f H(f^{(k)}, \mathcal{W}) \, \widehat{\bar{f}\!-\!\! f^{(k)}} }_2^2.  \]

The loss is taken through all depths $0 \leq k \leq D-1$, therefore the overall objective is defined as:
\nopagebreak
\begin{IEEEeqnarray}{cc}
\mathscr{O} (g, \bar{f}, \mathcal{W}) \defeq \sum_{k=0}^{D-1} \bigg( \gamma_a^k \, & \cdot \dfloor[\big]{ H(f^{(k)}, \mathcal{W}) - \bar{f}\,}_2^2 + \gamma_s \dfloor[\big]{ \Rad\big( \mathcal{C}_\mathcal{W}(L^s(f^{(k)})) \big) }_2^2 + \nonumber\\
& \negmedspace {} + \gamma_g \dfloor[\big]{ \nabla_f H(f^{(k)}, \mathcal{W}) \, \widehat{\bar{f}\!-\!\! f^{(k)}} }_2^2 \bigg)
\label{eq:iterative_loss}
\end{IEEEeqnarray}

The generalised error optimisation problem becomes:
\begin{IEEEeqnarray}{r'l}
	\min_\wei & \meanlim_{(g, \bar{f}) \sim \mathcal{D}_{g\times \bar{f}}} \left[ \mathscr{O} (g, \bar{f}, \mathcal{W}) \right],
\end{IEEEeqnarray}
where $\mathcal{D}_{g\times \bar{f}}$ is the joint distribution of realistic, noisy, undersampled sinograms alongside the ideal, high-quality reconstructions.

\section{Experimentation} \label{sec:unrolled_experimentation}

In this section we gather our experimentation framework and results.

\subsection{Architecture} \label{subsec:architecture}

In \ref{sec:unrolled_design} the major architecture was already introduced. It is briefly outlined here again with a caveat. We also depicted it in Fig.~\ref{fig:iterative_model}.

Let $L_\omega$ denote one step of the Landweber-iteration~\eqref{it:Landweber}:
\[ L_\omega(f) = f + \omega \Rad^*(g-\Rad f), \]
where the necessary and sufficient condition for convergence was provided by \eqref{eq:condition_Landweber}: $0 < \omega < 2/\sigma_{\text{max}}^2 = 2/ \norm[2]{\Rad^*\Rad}$.
Let us have a residual neural network $F_\mathcal{C}(\cdot, \mathcal{W}) = \I + \mathcal{C}(\cdot, \mathcal{W})$ with parameter vector $\mathcal{W}$. In this case the whole system operator becomes the composition of $s$ Landweber iterations and one call to the network:
\[ H(f, \mathcal{W}) = F_\mathcal{C}\left( L_\omega^s(f), \mathcal{W} \right). \]

For $\mathcal{C}$ we adapt the modified U-Net structure presented in \ref{subsec:UNet} and depicted in Fig.~\ref{fig:unet_Huang}. Our adaptation of the design concerns the up-sampling implementation. A 2-strided, 2x2 kernel-sized transposed convolution is preferred over the usually more well-behaved choice of standard interpolation based upsizing accompanied by a 2x2-convolution. The reason for that is mostly technical. For error backpropagation the derivative of the objective function~\eqref{eq:iterative_loss} would contain the term $\nabla^2_{\mathcal{W},f}H(f^{(k)}, \mathcal{W})$, i.e.\ all second order derivatives of every operation performed in the network would be required. However, the automatic differentiation framework TensorFlow \cite{MAP15} does not possess a second order derivative for resize operations with any interpolation.

\TODOm{power method for ART}

\subsection{Dataset. Preprocessing the data}

The dataset and its main traits have already been presented in subsection~\ref{subsec:data_preprocess}.

The LIDC-IDRI (\cite{LIDC11, LIDC11data, LIDC11dataweb}) dataset consists of approximately $250$ thousand $512 \times 512$ reconstructions belonging to $1010$ patients. Outside the training loop, offline, these images are downsampled to $128 \times 128$ mainly due to memory limitation. Afterwards, their $40$-projection Radon-transforms are produced.

For training, the $128 \times 128$ downsampled reconstruction is used as ideal, expected output reconstruction, the persisted sinograms represent the ideal, expected output sinogram. For the input of the neural network, $p=6$ ART-steps are performed on a $\underline{0}$ hypothesis using the noise contaminated sinogram. The noise generated on the sinogram is based on the description in \eqref{eq:noise_model}. Again, parameters $\mathcal{I}_0$ and $\sigma$ are set such that the SNR of noise on sinograms is around $40$ dB.

\subsection{Hyper-parameter setting}
As already stated, initiation is done with $p=6$ Landweber-iterations. Learning was unrolled to the final depth $D=4$. The amplifier for the fidelity of iterated reconstructions was chosen $\gamma_a = 2.0$, the multiplier of the support space error is $\gamma_s = 0.01$, the multiplier of the directional input gradient regularisation is $\gamma_g = 0.03$.

\subsection{Results}

Once again, as in subsection~\ref{subsec:results_convnet}, the test dataset consists of entire patient datasets, which have never been processed by the network during training and validation iterations. For the interpretation of metrics, the reader is reminded about the content of subsection~\ref{subsec:results_convnet}.

\subsubsection{Evaluation and comparison}

The current method is evaluated and compared to the previous method described in Chapter~\ref{chap:convnet}. The results are visible in Table~\ref{tab:results_iterative}. The rows contain results for a single metric. The metrics used are the mean absolute error expressed in Hounsfield Units (\emph{MAE [HU]}), the structural similarity index measure (\emph{SSIM}), the signal to noise ratio expressed in dB (\emph{SNR [dB]}) and the relative error between the computed and ideal reconstructions (\emph{RelError}). We displayed the evaluation for three methods. The first column shows metrics for the standard FBP algorithm. The middle column repeats the results for our method presented in the previous chapter, which is a post-processing type residual neural network, called measurement-consistent, sparsifying postprocess-ConvNet (\emph{MC-S-P-ConvNet}). Last, but not least, the last column enlists results for our most recent method, the unrolled support-kernel iterative regulariser GD (\emph{USKI-R-GD}).

The numerical comparison shows that the iterative usage of neural networks has significant room for improvement, the standard post-processing neural regularisation could not be outperformed. There is a relatively straightforward explanation for this negative result. In the iterative case the neural network is continuously fed with data from different depths of iteration and it is supposed to model a meaningful regression for all iterations. This probably calls for the expansion on the hypothesis space, i.e.\ a network with more parameters should be taken. Nevertheless, enlarging the network often leads to \emph{overfitting}, a situation, where the neural network starts learning hidden features that are characteristic for the training set and not general. Parallel to overfitting, the generalisation performance and error of the system usually worsen.

\begin{table}
\centering
\begin{tabular}{cccc}
\hline
& FBP & MC-S-P-ConvNet & New: USKI-R-GD \\
\hline
MAE [HU] & 272.9 & 38.8 & 62.97 \\
SSIM & 0.3624 & 0.924 & 0.87 \\
SNR [dB] & 6 & 20.5 & 19.88 \\
RelError & 0.53 & 0.1 & 0.15\\
\hline
\end{tabular}
\caption{Here we present our results on simple FBP algorithm, the previously developed FBP + ConvNet solution and our algorithm.}
\label{tab:results_iterative}
\end{table}

\subsubsection{Case-studies} \label{subsec:case_study:iterative}

Figures \ref{fig:casestudy_iterative_00} to \ref{fig:casestudy_iterative_09} present some case-studies. On each of them we have the following displayed: the original, ideal reconstruction; the initialised input of the iterative system, initialised with $p=6$ Landweber-steps; the first $4$ iterates. On the right hand-side we depicted the cross-section of the $4^\text{th}$ iterate at $y=64$ and compared it to the ideal reconstruction. Below that metrics related to the sample are outlined.

It is easily noted that the quality of reconstructions improved with the depth of iteration, which gives an affirmative answer to whether neural networks are capable of iterative refinement. Even though the generalisation performance of the method falls behind the post-processing algorithm, still it was not straightforward that iterative improvement in the error is possible.

\newcommand{\myfigb}[5][0.20]{
	\begin{minipage}{#1\textwidth}
		\centering
		\subfloat[#5]{
			\includegraphics[width=\textwidth,keepaspectratio]
			{images/IARTRN128_inputgrad3e-2_kernerr1e-2_ampl2.0-rec_HU_mae-63.719-20210519-085621/#2/#3.png}%
			\label{fig:casestudy_#2:#4}
		}
	\end{minipage}
}

\newcommand{\iterativecasestudy}[5]{
	\begin{figure}[ht]
		\centering
		\myfigb{#1}{00_expected_reconstruction}{a}{Expected rec.}
		\myfigb{#1}{00_init_6_ART}{b}{Init:$6 \times L$}
		\myfigb{#1}{Init_+_1x_ART+ResNet_}{c}{$\text{Init} + 1\times H$}
		\myfigb[0.28]{#1}{crossection_y64}{cross}{Cross-section, $y=64$}
		
		\myfigb{#1}{Init_+_2x_ART+ResNet_}{d}{$\text{Init} + 2\times H$}
		\myfigb{#1}{Init_+_3x_ART+ResNet_}{e}{$\text{Init} + 3\times H$}
		\myfigb{#1}{Init_+_4x_ART+ResNet_}{f}{$\text{Init} + 4\times H$}
		\begin{minipage}{0.28\textwidth}
			\centering
			\begin{minipage}{0.8\textwidth}
				\centering
				\small
				\begin{tabular}{cc}
					\hline
					MAE [HU] & #2 \\
					SSIM & #3 \\
					SNR & #4 \\
					RelErr & #5 \\
					\hline
				\end{tabular}
				\captionof{figure}{}
				\label{fig:casestudy_iterative_#1}
			\end{minipage}
		\end{minipage}%
	\end{figure}%
}

\iterativecasestudy{00}{80.01}{0.830}{16.68}{0.147}
\iterativecasestudy{02}{61.41}{0.870}{17.16}{0.139}
\iterativecasestudy{03}{64.70}{0.859}{17.05}{0.140}
\iterativecasestudy{06}{70.32}{0.861}{16.73}{0.146}
\iterativecasestudy{07}{55.21}{0.885}{17.95}{0.127}
\iterativecasestudy{08}{75.87}{0.820}{16.28}{0.153}
\iterativecasestudy{09}{75.12}{0.846}{14.94}{0.179}

\subsubsection{Semi-convergence}

When having analysed the Landweber-iteration, in subsection~\ref{subsec:Landweber}, we derived the reason for an empirical effect, called \emph{semi-convergence}. We were curious to see if the learning approach hid a similar phenomenon. For this we evaluated the first $6$ iterations $H$ for a randomly chosen set of $100$ reconstructions and took the mean between the mean absolute error expressed in HU of reconstructions at the same iteration level. The process was repeated for the relative error and SSIM. The results are visible in Fig.~\ref{fig:semi-convergence-NN}. In separate case-studies this effect was already experienced and noted in subsection~\ref{subsec:case_study:iterative}, yet it is present in general. The explanation for why the minimum is at iteration $3$ or $4$ resides in the fact that the final iteration depth was chosen to be $4$.

\begin{figure}[ht]
	\centering
	\begin{minipage}{.3\linewidth}	
		\centering
		\subfloat[MAE HU]{
			\includegraphics[keepaspectratio,width=\linewidth]{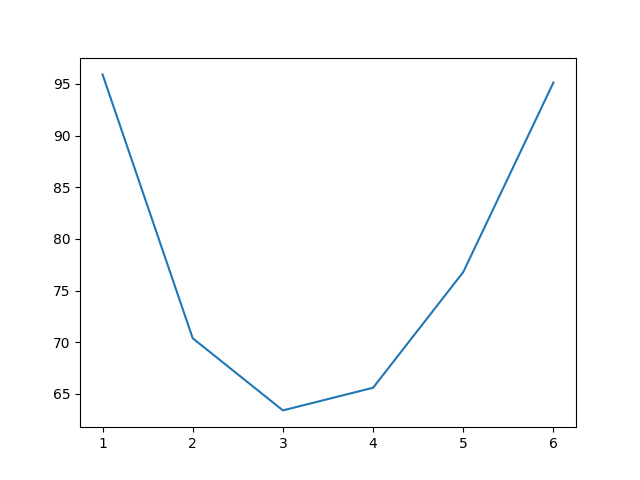}
			\label{fig:semi-convergence-NN:a}
		}
	\end{minipage}
	\begin{minipage}{.3\linewidth}
		\centering
		\subfloat[Relative Error]{
			\label{fig:semi-convergence-NN:b}
			\includegraphics[keepaspectratio,width=\linewidth]{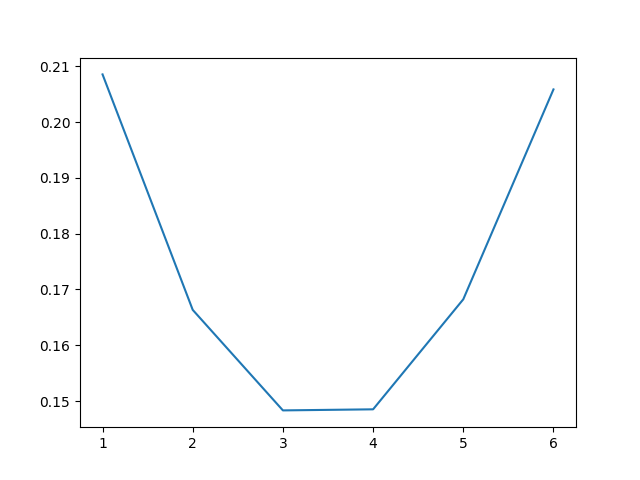}			
		}
	\end{minipage}
	\begin{minipage}{.3\linewidth}
		\centering
		\subfloat[SSIM]{
			\label{fig:semi-convergence-NN:c}
			\includegraphics[keepaspectratio,width=\linewidth]{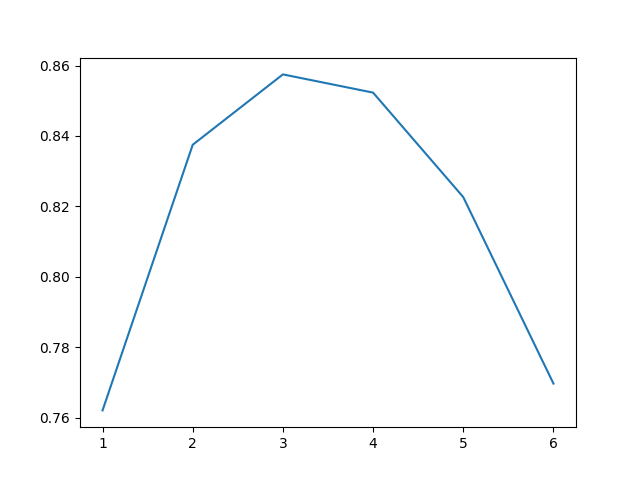}			
		}
	\end{minipage}
	
	\caption{\subref{fig:semi-convergence-NN:a} The semi-convergence in terms of mean absolute error expressed in HU. \subref{fig:semi-convergence-NN:b} The semi-convergence in terms of relative error. \subref{fig:semi-convergence-NN:c} The semi-convergence in terms of SSIM.}
	\label{fig:semi-convergence-NN}
\end{figure}

\chapter{Conclusions} \label{chap:conclusion}

The results of our report are two-folded. First, we presented an architecture based on the fully convolutional U-Net, but including novel elements. The suggested model incorporates a data consistency module as well as other objective functions borrowed from the field of compressive sampling. Therefore, our architecture is capable of an end-to-end training process that was previously emulated in an alternating fashion by training a target network and then modifying based on all other constraints. We believe that this is a fruitful direction.

Currently the most prevailing obstacle impeding us for further reducing reconstruction error is that our sparsity measuring operator, the total variation optimiser loss does not have its minimum where all other losses have theirs. In fact, the logarithmic operator defined in~\eqref{eq:log_min} does not yield any minimum, its use is rather heuristic and its scope in the training process is purely regularisation. Therefore, we strongly believe that the learning model's weights start to oscillate around its optimum value. In our upcoming research we want to dedicate time for studying different sparsity enhancing operators and tuning their parameters optimally. For instance, the weighting of the total variation minimiser in our project is bounded from both directions. On one hand, a high weight factor most probably causes an oscillation with higher amplitude. On the other hand, a reduced coefficient would result in mitigated regularisation effects. Other possible sparsity operators may include the ones presented in Chapter~\ref{chap:preliminaries}, particularly the NLTV operator defined as~\eqref{eq:NLTV-def}.

In an other attempt we presented a fundamentally new concept by combining iterative methods and neural networks. The results are yet to be improved to become state-of-the-art. The methodology's main take-away is that the  effect called \emph{semi-convergence} is reproducible even for neural systems. This meant that a CNN taught to reconstruct the solution of an inverse problem iteratively was capable of improving its own reconstruction with further iterations, though not going further than the design depth. Our future plan is to analyse more, how these regularisation terms could be improved.

Furthermore, new metrics should be defined to measure the performance of the system. This report lacks the analysis of reconstructions that originally displayed cancerous tumours. It has to be assessed whether such lesions stay intact on images and it would also be interesting to see if the SNR around these lesions increases after applying our model to the inputted noisy reconstruction. Besides that it is desirable to design a metric that would be sensitive to changes in regions where ideally relatively large intensity differences occur, since these could hold diagnostic information.


\newpage
\phantomsection\addcontentsline{toc}{chapter}{Acknowledgements}
\section*{Acknowledgements}
This work was supported by the Ministry of Human Capacities under its New National Excellence Program (ÚNKP-19-2-I-BME-354 and ÚNKP-20-2-I-BME-117) and under its Human Capacity Development Program (grant EFOP 3.6.1.- 16-2016-00014 with title \begin{otherlanguage}{hungarian}``Diszruptív technológiák kutatás-fejlesztése az e-mobility területén és integrálásuk a mérnökképzésbe''\end{otherlanguage}).

Here I would like to thank my supervisor, D\'{a}niel Hadh\'{a}zi for his patience and everlasting support in our work.

\appendix

\bibliographystyle{apalike}
\bibliography{Bibliography/Medical_Computing,Bibliography/Neural_Networks_in_Medical_Imaging,Bibliography/Webpages}

\end{document}